\documentclass{article}
\usepackage{authblk}
\PassOptionsToPackage{numbers, compress}{natbib}
\usepackage{subcaption}
\usepackage{microtype}
\usepackage{graphicx}
\usepackage[margin=1in]{geometry}
\usepackage{booktabs} 
\usepackage[hidelinks]{hyperref}
\usepackage[numbers,authoryear,compress,sort]{natbib}

\usepackage{algorithm}
\usepackage{algorithmic}

\usepackage{tablefootnote}

\usepackage{amsmath}
\usepackage{amssymb}
\usepackage{mathtools}
\usepackage{amsthm}
\usepackage{thm-restate}
\usepackage{mathrsfs}
\usepackage[inline]{enumitem}
\usepackage{multirow}

\usepackage[capitalize,noabbrev]{cleveref}
\usepackage{dp_ssl_preset}

\theoremstyle{plain}

\newtheorem{lemL}{Lemma}

\usepackage[most]{tcolorbox}
\usepackage[textsize=tiny]{todonotes}
\usepackage{bbm}
\usepackage[utf8]{inputenc}
\usepackage{nicefrac}
\usepackage{tikz}

\usepackage{array}
\newcolumntype{L}[1]{>{\raggedright\let\newline\\\arraybackslash\hspace{0pt}}m{#1}}
\newcolumntype{C}[1]{>{\centering\let\newline\\\arraybackslash\hspace{0pt}}m{#1}}
\newcolumntype{R}[1]{>{\raggedleft\let\newline\\\arraybackslash\hspace{0pt}}m{#1}}

\usetikzlibrary{shapes,decorations,arrows,calc,arrows.meta,fit,positioning}
\tikzset{
    -Latex,auto,node distance =1 cm and 1 cm,semithick,
    state/.style ={ellipse, draw, minimum width = 0.7 cm},
    point/.style = {circle, draw, inner sep=0.04cm,fill,node contents={}},
    bidirected/.style={Latex-Latex,dashed},
    el/.style = {inner sep=2pt, align=left, sloped}
}
\usepackage{array}
\newcolumntype{L}[1]{>{\raggedright\let\newline\\\arraybackslash\hspace{0pt}}m{#1}}
\newcolumntype{C}[1]{>{\centering\let\newline\\\arraybackslash\hspace{0pt}}m{#1}}
\newcolumntype{R}[1]{>{\raggedleft\let\newline\\\arraybackslash\hspace{0pt}}m{#1}}

\usepackage{amartya_ltx}

\title{Provable Privacy with Non-Private Pre-Processing}
\author{Yaxi Hu \thanks{yaxi.hu@tuebingen.mpg.de}}
\author{Amartya Sanyal \thanks{amsa@di.ku.dk}}
\author{Bernhard Sch\"olkopf\thanks{bernhard.schoelkopf@tuebingen.mpg.de}}
\date{}
\affil{Max Planck Institute for Intelligent Systems, T\"ubingen, Germany}

\begin{document}
\maketitle

\begin{abstract}
When analyzing Differentially Private~(DP) machine learning pipelines, the potential privacy cost of data-dependent pre-processing is frequently overlooked in privacy accounting. In this work, we propose a general framework to evaluate the additional privacy cost incurred by non-private data-dependent pre-processing algorithms. Our framework establishes upper bounds on the overall privacy guarantees by utilising two new technical notions: a variant of DP termed Smooth DP and the bounded sensitivity of the pre-processing algorithms. In addition to the generic framework, we provide explicit overall privacy guarantees for multiple data-dependent pre-processing algorithms, such as data imputation, quantization, deduplication, standard scaling and PCA, when used in combination with several DP algorithms. Notably, this framework is also simple to implement, allowing direct integration into existing DP pipelines.
\end{abstract}

\section{Introduction}
\label{sec:intro}

With the growing emphasis on user data privacy, Differential Privacy (DP), has become the preferred solution for safeguarding training data in Machine Learning~(ML) and data analysis \citep{dwork06dp}. DP algorithms are designed to ensure that individual inputs minimally affect the algorithm's output, thus preserving privacy. This approach is now widely adopted by various organizations for conducting analyses while maintaining user privacy \citep{appleldp,USCensusDP}.

Pre-processing data is a standard practice in data analysis and machine learning. Techniques such as data imputation for handling missing values \citep{imputation19}, deduplication for reducing memorization and eliminating bias \citep{kandpal22dedup, lee-etal-2022-deduplicating}, standard scaling for reducing the impact of outliers, and dimensionality reduction for denoising or visualization \citep{Abadi16dpsgd,zhou21pcasgd,pinto2023pillar} are commonly used. These pre-processings are also prevalent prior to applying DP algorithms, a pipeline that we term~\emph{pre-processed DP pipeline}. Among other uses, this has been shown to improve privacy-accuracy trade-off \citep{tramer2021BetterFeatures, Ganesh23why-public-pretraining}.

A fundamental assumption, required by DP, is that individual data points are independent. However, if training data is used in pre-processing, this assumption is compromised. For example, when deduplicating a dataset, whether a point remains after the pre-processing is dependent on the presence of other points in its vicinity. Similarly, for mean imputation, the imputed value depends on the values of other data points. These dependencies, also evident in PCA and pre-training, can undermine the privacy guarantee of the pre-processed DP pipeline.

There are multiple strategies to address this. A straightforward method to derive privacy guarantees for this pipeline is to use group privacy where the size of the group can be as large as the size of the dataset, thereby resulting in weak privacy guarantees. This idea, albeit not in the context of pre-processing, was previous explored under the name Dependent Differential Privacy~(DDP)~\citep{Zhao17ddp,liu2016ddp}. Another approach is to use public data for the pre-processing. Broadly referred to as semi-private learning algorithms, examples of these methods include pre-training on public data and learning projection functions using the public data e.g. \cite{pinto2023pillar, li2021large, li2022private, yu2021do}. Despite their success, these methods crucially rely on the availability of high-quality public data.

In the absence of public data, an alternative approach is to privatise the pre-processing algorithm. However, designing new private pre-processing algorithms complicate the process, increasing the risk of privacy breaches due to errors in implementation or analysis. Moreover, private pre-processing can be statistically or computationally more demanding than private learning itself. For example, the sample complexity of DP-PCA\citep{chaudhuri2012dppca} is dependent on the dimension of the training data \citep{Liu22dp-pca}, implying that the costs associated with DP-PCA could surpass the benefits of private learning in a lower-dimensional space.

To circumvent these challenges, an alternative approach is non-private pre-processing with a more rigorous analysis of the entire pipeline. This method is straightforward, circumvents the need for modifying existing processes, and avoids the costs associated with private pre-processing. Naturally, this raises the important question, 

\begin{quote}
    \emph{What is the price of non-private pre-processing in differentially private data analysis ?}
\end{quote}  

Our work shows that the overall privacy cost of pre-processed DP pipeline can be bounded with minimal degradation in privacy guarantee. To do this, we rely on two new technical notions: sensitivity of pre-processing functions~(\Cref{defn:sensitivity-pre}) and Smooth-DP~(\Cref{defn:smoothness}). In short, 

\begin{enumerate}
\vspace{-3mm}
\item We introduce a generic framework to quantify the privacy loss in the \emph{pre-processed DP pipeline}. Applying this framework, we evaluate the impact of commonly used pre-processing techniques such as deduplication, quantization, data imputation, standard scaling, and PCA on the overall privacy guarantees.
\item We base our analysis on a novel variant of differential privacy, termed smooth DP, which may be of independent interest. We demonstrate that smooth DP retains essential DP characteristics, including post-processing and composition capabilities, as well as a slightly weaker form of amplification by sub-sampling.
\item We propose an algorithm to balance desired privacy levels against utility in the \emph{pre-processed DP pipeline}. This approach is based on the Propose-Test-Release Mechanism, allowing the user to choose desired privacy-utility trade-off.
\end{enumerate}

\paragraph{Related Work}
\label{sec:related}
Closely related to our work,~\citet{debenedetti2023sidechannels} also studies the necessity of conducting privacy analysis across the entire ML pipeline, rather than focusing solely on the training process. They identify that pre-processing steps, such as deduplication, can unintentionally introduce correlations into the pre-processed datasets, leading to privacy breaches. They show that the empirical privacy parameter of DP-SGD~\citep{Abadi16dpsgd} can be deteriorated by over five times with membership inference attack designed to exploit the correlation introduced by deduplication. While they show privacy attacks using deduplication with DP-SGD that can maximise the privacy loss, our work quantifies the privacy loss in deduplication as well as other pre-processing algorithms with several privacy preserving mecahnisms, thereby presenting a more holistic picture of this problem.

Another line of research, focusing on privacy in correlated datasets~\citep{liu2016ddp, Zhao17ddp, Humphries2023miadependent}, shows that correlations in the datasets can increase privacy risk of ML models. In response, Pufferfish Differential Privacy~\citep{kifer14pufferfish,song17pufferfish} and Dependent Differential Privacy \cite{liu2016ddp, Zhao17ddp} were proposed as privacy notions tailored for datasets with inherent dependencies. However, these definitions usually require complete knowledge of the datasets' dependency structure or the data generating process and sometimes lead to vacuous privacy guarantees. Moreover, the application of these privacy notions usually complicates the privacy analysis, as many privacy axioms, such as composition, do not hold under these more general notions. 

\section{Preliminaries in Differential Privacy}
\label{sec:prelim}
Before providing the main results of our work, we first introduce some common definitions and mechanisms in Differential Privacy. 

\subsection{Rényi Differential Privacy}

Differential privacy restricts the change in the output distribution of a randomized algorithm by altering a single element in the dataset. Formally, for $\varepsilon, \delta > 0$, a randomized algorithm $\cA$ satisfies $(\varepsilon, \delta)$-DP if for any two dataset $S, S'$ that differ by exactly one element and any possible output of the algorithm $\cO$, \begin{equation}\label{eq:approx-dp}\bP\bs{\cA(S) \in \cO }\leq e^\varepsilon \bP\bs{\cA(S') \in \cO} + \delta. \end{equation}
In this work, we mainly focus on a stronger notion of DP, known as Rényi Differential Privacy (RDP)\citep{mironov2017renyi}, based on the Rényi divergence between two distributions. 
\begin{defn}[Rényi Divergence]\label{defn:renyi-divergence}
Let $P, R$ be two probability distributions with $\operatorname{supp}(P)\subseteq \operatorname{supp}(R)$. Let $\alpha > 1$. The Rényi divergence with order $\alpha$ between $P$ and $R$ is defined as \[D_\alpha(P||R) = \frac{1}{\alpha - 1}\log \bE_{x\sim R}\br{\frac{P(x)}{R(x)}}^\alpha.\]
\end{defn}

Let $d_H(\cdot, \cdot)$ denote the Hamming distance between two datasets, we formally define RDP as follows. 
\begin{defn}[$(\alpha, \varepsilon(\alpha))$-RDP]\label{defn:rdp-over-sets}
Let $\varepsilon(\alpha)$ be a function that maps each $\alpha$ to a positive real number. A randomized algorithm $\cA$ is $(\alpha, \varepsilon(\alpha))$-RDP if for all $\alpha > 1$, for any datasets $S, S'$ differing at a single point, it holds that \[D_\alpha(\cA(S)||\cA(S')) \leq \varepsilon(\alpha).\]
\end{defn}

This definition of RDP can be easily converted to standard DP via~\Cref{lem:rdp-to-adp}. While~\Cref{defn:rdp-over-sets} is unconditional on the possible set of datasets, a relaxed version of conditional RDP can be defined over a given dataset collection $\cL$ such that the neighboring datasets $S, S'\in\cL$. When the dataset collection is known in advance, the conditional definition of RDP allows for tighter privacy analysis. We use this conditional version in some of our analysis\footnote{To avoid confusion, we note that our privacy guarantees are not \emph{conditional}, in the sense that it does not suffer a catastrophic failure under any dataset.}.

\subsection{Private mechanisms}\label{sec:private-mechanisms}

There are several ways to make non-private algorithms  private. All of them implicitly or explicitly add carefully calibrated noise to the non-private algorithm. Below, we briefly define the three most common ways in which DP is injected in data analysis tasks and machine learning algorithms.

\paragraph{Output perturbation} The easiest way to inject DP guarantees in an estimation problem is to perturb the output of the non-private estimator with appropriately calibrated noise. Two most common ways to do so are \textit{Gaussian Mechanism} and \textit{Laplace Mechanism}. For any deterministic estimator $f$, both mechanisms add noise proportional to the global sensitivity $\Delta_f$, defined as the maximum difference in $f$ over all pairs of neighboring datasets. For a given privacy parameter $\varepsilon > 0$, both the Gaussian mechanism, denoted $\cM_G$, and the Laplace mechanism, denoted $\cM_L$, produce an output of the form $f(S) + \xi$. Here, $\xi$ follows a Gaussian distribution $\cN(0, \nicefrac{\Delta_f^2}{\varepsilon^2})$ for the Gaussian mechanism, and a Laplace distribution $\operatorname{Lap}(\nicefrac{\Delta_f}{\varepsilon})$ for the Laplace mechanism. 

\paragraph{Random sampling} While Output perturbation is naturally suited to privatising the output of non-private estimators, it is less intuitive when selecting discrete objects from a set. In this case, a private mechanism can sample from a probability distribution defined on the set of objects. The \textit{Exponential Mechanism}, denoted as $\cM_E$, falls under this category and is one of the most fundamental private mechanisms. Given a score function $Q$ with global sensitivity $\Delta_Q$, it randomly outputs an estimator $w$ with probability proportional to $\exp\br{\frac{Q(w, S)\varepsilon}{2\Delta_Q}}$. 

\paragraph{Gradient perturbation} Finally, most common ML applications use gradient-based algorithms to minimize a loss function on a given dataset. A common way to inject privacy in these algorithms is to introduce Gaussian noise into the gradient computations in each gradient descent step. This is referred to as Differential Private Gradient Descent (DP-GD) denoted as $\cA_{\mathrm{GD}}$ \citep{bassily2014private,song21Evading}. Other variants that are commonly used are Differential Private Stochastic Gradient Descent (DP-SGD) with subsampling \citep{bassily2014private,Abadi16dpsgd}, denoted as $\cA_{\mathrm{SGD-samp}}$, and DP-SGD with iteration \citep{feldman18iteration}, denoted as $\cA_{\mathrm{SGD-iter}}$. We include the detailed description of each of these methods in~\Cref{app:gradient-based-methods}.

\section{Main Results}
\label{sec:pre-processing}

We first introduce a norm-based privacy notion, called Smooth RDP, that allows us to conduct a more fine-grained analysis on the impact of pre-processing algorithms. Using this definition, we establish our main results on the privacy guarantees of a pre-processed DP pipeline. 

\subsection{Smooth RDP}\label{sec:srdp}

\begin{table*}[t]\small
\caption{RDP and SRDP parameters of DP mechanisms. We let $\gamma$ and $\maxD$ denote the inverse pointwise divergence and maximum divergence between two datasets.}
    \centering
    \begin{tabular}{C{0.075\textwidth}C{0.1\textwidth}C{0.1\textwidth}C{0.4\textwidth}C{0.05\textwidth}C{0.18\textwidth}}
        \toprule
        Notation & Meaning & Mechanism  & Assumptions & RDP & SRDP \\
        \midrule
       \multirow{2}{*}{$f$} & \multirow{2}={Output function} & $ \cM_G $&$f$ is $L$-Lipschitz &$\frac{\alpha\varepsilon^2}{2}$&$\frac{\alpha L^2\sdpL^2\varepsilon^2}{2\Delta_f^2}$\\
        & & $\cM_L$ & $f$ is $L$-Lipschitz &$ \varepsilon$&$ \frac{ L \sdpL}{\Delta_f}\varepsilon$\\\midrule
        $Q$& Score function &$\cM_E$& $Q$ is $L$-Lipschitz &$ \varepsilon$&$ \frac{L \sdpL}{\Delta_Q}\varepsilon$\\\midrule

       $\ell$& Loss function &$\cA_{\mathrm{GD}}$ & $\ell$ is $L$-Lipschitz and $\mu$-smooth, $\sigma = \frac{L\sqrt{T}}{\varepsilon n}$ & $ {2\alpha\varepsilon^2}$& $ \frac{\alpha\mu^2\sdpL^2\varepsilon^2}{2L^2}$\\
       \addlinespace
        $T$& Number of iterations &$\cA_{\mathrm{SGD-samp}}$ &$\ell$ is $L$-Lipschitz and $\mu$-smooth, $\sigma = \Omega(\nicefrac{L\sqrt{T}}{\varepsilon n })$, inverse point-wise divergence $\gamma$, $1\leq \alpha\leq \min\bc{\frac{\sqrt{T}}{\varepsilon}, \frac{L^2T}{\varepsilon^2 n^2}\log\frac{n^2\varepsilon}{L\sqrt{T}}}$&$\frac{\alpha^2\varepsilon^2}{2}$ & $ \frac{\alpha\mu^2\sdpL^2\varepsilon^2\gamma^2}{2L^2}$\\
        \addlinespace
         $\sigma$ & Variance of gradient noise & \multirow{2}{*}{$\cA_{\mathrm{SGD-iter}}$} & \multirow{2}={$\ell$ is convex, $L$-Lipschitz and $\mu$-smooth, $\sigma = \frac{8\sqrt{2\log n}\eta L}{\varepsilon\sqrt{n}}$, $\varepsilon = O(1/n\alpha^2)$, maximum divergence $\maxD$, $L\sqrt{2\alpha(\alpha -1)}\leq \sigma$ }
         & \multicolumn{1}{c}{\multirow{2}{*}{$\frac{\alpha\varepsilon^2}{2}$}}&\multicolumn{1}{c}{\multirow{2}{*}{$\frac{\alpha\sdpL^2\mu^2 n\log(n-\maxD+2)}{2(n-\maxD+1)L^2\log n}$}}\\
         \addlinespace
       $\eta$&Learning rate&&&\\
        \bottomrule
    \end{tabular}
    \label{tab:sdp-parameters}
\end{table*}

Our analysis on the privacy guarantees of pre-processed DP pipelines relies on a privacy notion that ensures indistinguishability between two datasets with a bounded $L_{12}$ distance. Here, the $L_{12}$ distance between two datasets $S$ and $S'$ of size $n$ is defined as $d_{12}(S, S') = \sum_{i = 1}^n \norm{S_i - S_i'}_2$. We introduce this privacy notion as Smooth Rényi Differential Privacy (SRDP), defined as follows:

\begin{defn}[$(\alpha, \varepsilon(\alpha, \sdpL))$-smooth RDP]\label{defn:srdp}
    Let $\varepsilon(\alpha, \sdpL)$ be a function that maps each $\alpha, \sdpL$ pair to a real value. A randomized algorithm $\cA$ is $(\alpha, \varepsilon(\alpha, \sdpL))$-SRDP if for each $\alpha > 1$ and $\sdpL > 0$, \[\sup_{\substack{S, S':d_{12}(S, S')\leq \sdpL}}D_\alpha(\cA(S)||\cA(S'))\leq \varepsilon(\alpha, \sdpL).\]
\end{defn}

SRDP shares similarities with distance-based privacy notions \citep{Lecuyer19pixeldp, 23distancedp}, but they differ in a key aspect: the distance-based privacy considers neighboring datasets with bounded distance, while SRDP allows for comparison over two datasets differing in every entry. 

Similar to conditional RDP over a set $\cL$, we define conditional SRDP over a set $\cL$ by imposing the additional assumption that $S, S'\in \cL$ in~\Cref{defn:srdp}.

While conditional SRDP over a set $\cL$ can be considered as a special case of Pufferfish R\'enyi Privacy~\citep{kifer14pufferfish}, we show that it satisfies desirable properties, such as sequential composition and privacy amplification by subsampling, which are not satisfied by the more general Pufferfish Renyi Privacy \citep{Pierquin2023renyi}. 

\paragraph{Properties of SRDP} Similar to RDP, SRDP satisfies (sequential) composition and closure under post-processing. 

\begin{restatable}{lem}{propertiesSRDP}\label{lem:properties-srdp}
   Let $\cL$ be any dataset collection. Then, the following holds.
   \begin{itemize}[leftmargin= 0.8em]
       \item\textbf{Composition} Let $k$ be a positive integer. Let $\alpha > 0, \varepsilon_i: \reals\times \reals \to \reals, \forall i\in [k]$. For any $i\in [k]$, if the randomized algorithms $\cA_i$ is $(\alpha, \varepsilon_i(\alpha, \sdpL))$-SRDP over $\cL$, then the composition of the $k$ algorithms $(\cA_1, \ldots, \cA_k)$ is $(\alpha, \sum_{i = 1}^k \varepsilon_i(\alpha, \sdpL))$-SRDP over $\cL$. 
       \item\textbf{Post-processing} Let $\alpha > 0, \varepsilon: \reals \times \reals \to \reals$, and $f$ be an arbitrary algorithm. For any $\sdpL > 0$, if $\cA$ is $(\alpha, \varepsilon(\alpha, \sdpL))$-SRDP over $\cL$, then $f\circ\cA$ is $(\alpha, \varepsilon(\alpha, \sdpL))$-SRDP over $\cL$. 
   \end{itemize} 
\end{restatable}

SRDP also satisfies a form of Privacy amplification by subsampling. We state the weaker version without additional assumptions in~\Cref{app:sec2-proofs} and use a stronger version for $\cA_{\mathrm{SGD-iter}}$ in~\Cref{thm:sdp-parameters} with some light assumptions.

\paragraph{SRDP parameters for common private mechanisms}  In~\Cref{thm:sdp-parameters}, we present the RDP and SRDP parameters for the private mechanisms discussed above. The SRDP parameter usually relies on the Lipschitzness and Smoothness~(see~\Cref{app:sec2-proofs} for definitions) of the output or objective function.

\begin{restatable}[Informal]{thm}{SDPtable}\label{thm:sdp-parameters}
    The DP mechanisms discussed in~\Cref{sec:private-mechanisms} satisfy RDP and SRDP under assumptions on the output or the objective functions. We summarize the parameters and their corresponding assumptions in~\Cref{tab:sdp-parameters}. 
\end{restatable}

\Cref{tab:sdp-parameters} demonstrates that the RDP parameter $\varepsilon$ of most private mechanisms increases to $O(\sdpL\varepsilon)$ for SRDP. In contrast, a naive analysis using group privacy inflates the privacy parameter to $O(n\varepsilon)$, as any two datasets with $d_{12}$ distance $\sdpL$ can differ by at most $n$ entries. Hence, SRDP always leads to tighter privacy parameter than group privacy for $\sdpL = o(n)$, which we are able to exploit later. 

For DP-SGD by subsampling and iteration, the SRDP guarantee is dependent on two properties: the inverse pointwise divergence and maximum divergence of the datasets, detailed in~\Cref{app:gradient-based-methods}. Briefly, the inverse pointwise divergence $\gamma$ measures the ratio between the $d_{12}$ distance and the maximum pointwise distance between two datasets. The maximum divergence $\maxD$ measures the number of points that the two datasets differ. DP-SGD with subsampling benefits from a small $\gamma$, while DP-SGD with iteration benefits from a small $\maxD$. As we will show in the later sections, at least one of these conditions are usually met in practice.

\subsection{Privacy of Pre-Processed DP Pipelines}

Before stating the main result, we introduce the term \textit{data-dependent pre-processing algorithm}. Let $\cX\in \reals^d$ be the instance space with Euclidean norm bounded by 1, \ie $\forall x\in \cX$, $\norm{x}_2 \leq 1$. A deterministic \textit{data-dependent pre-processing algorithm} $\pi: \cX^n\to\cF$ takes a dataset $S$ as input and returns a \textit{pre-processing function} $\pi_S: \cX\to \cX\cup \emptyset$ in a function space $\cF$. Intuitively, privacy of a private algorithm is retained under non-private data-dependent pre-processing, if a single element in the dataset has bounded impact on the output of the pre-processing algorithm. We define the sensitivity of a pre-processing algorithm to quantify the impact of a simgle element below.

\begin{defn}\label{defn:sensitivity-pre}
    Let $S_1 = S\cup \{z_1\}$ and $S_2 = S\cup \{z_2\}$ be two arbitrary neighboring datasets, we define the $L_\infty$ and $L_2$ sensitivity of a pre-processing function $\pi$ as\footnote{When one of $\pi_{S_1}(x)$ and $\pi_{S_2}(x)$ is $\emptyset$, we define $\norm{\pi_{S_1}(x) - \pi_{S_2}(x)}_2 = 1$. } 
    \begin{equation}\label{defn:sensitivity-pre-processing}
    \begin{aligned}
            \Delta_\infty(\pi) &= \sup_{\substack{S_1, S_2: d_{H}(S_1, S_2)=1}}d_H(\pi_{S_1}(S), \pi_{S_2}(S)),\\ 
            \Delta_2(\pi) &= \sup_{\substack{S_1, S_2: d_{H}(S_1, S_2)=1}}\max_{x\in S}\norm{\pi_{S_1}(x) - \pi_{S_2}(x)}_2.
    \end{aligned}
\end{equation}
\end{defn}

When the neighboring datasets $S, S'$ are from a dataset collection $\cL$, we define conditional sensitivity of the pre-processing algorithm $\pi$ as $\Delta_\infty(\cL, \pi)$ and $\Delta_\infty(\cL, \pi)$ in a similar manner. The conditional sensitivity is non-decreasing as the size of the set of $\cL$ increases, i.e. if $\cL \subset \cL'$, then $\Delta_\infty(\cL, \pi)\leq \Delta_\infty(\cL', \pi)\leq \Delta_\infty(\pi)$ and $\Delta_2(\cL, \pi)\leq \Delta_2(\cL', \pi)\leq \Delta_2(\pi)$ for all $\pi$. 

In~\Cref{thm:general-thm-training-only}, we present the privacy guarantees of pre-processed DP pipeline in terms of the RDP and SRDP parameters of the private algorithm and the sensitivity of the pre-processing algorithms. This is a meta theorem, which we then refine to get specific guarantees for different combinations of pre-processing and private mechanisms in~\Cref{thm:overall-privacy-table2}. With a slight abuse of notation, we denote the output of a private algorithm $\cA$ on $\pi_S(S)$ as $\cA\circ \pi(S)$.  

\begin{restatable}{thm}{generalTrainingOnly}\label{thm:general-thm-training-only}
    For a set of datasets $\cL$ and for any $\alpha \geq 1$, $\sdpL > 0$, consider an algorithm $\cA$ that is $(\alpha, \varepsilon(\alpha))$-RDP and $(\alpha, \tilde{\varepsilon}(\alpha, \sdpL))$-SRDP over the set $\cL$. For a pre-processing algorithm $\pi$ with $L_\infty$ sensitivity $\Delta_\infty$ and $L_2$ sensitivity $\Delta_2$, $\cA\circ \pi$ is $\br{\alpha, \widehat{\varepsilon}}$-RDP over $\cL$ for all 
    \(c_1, c_2 \geq 1 \), where 
    \begin{equation}
        \label{eq:general-results}
        \begin{aligned}
            \widehat{\varepsilon} \leq \max\bigg\{ &\frac{\alpha c_1 - 1}{c_1\br{\alpha - 1}}\tilde{\varepsilon}\br{\alpha c_1, \Delta_2\Delta_\infty} + \varepsilon\br{\frac{c_1\alpha-1}{c_1-1}}, \\
            &\frac{\alpha c_2 - 1}{c_2(\alpha - 1)}\varepsilon\br{\alpha c_2 }+ \tilde{\varepsilon}\br{\frac{c_2\alpha - 1}{c_2-1},\Delta_2\Delta_\infty}\bigg\}.
        \end{aligned}
    \end{equation}     
\end{restatable}
\begin{proof}[Proof sketch]
Consider two neighboring datasets $S_1$ and $S_2$, where $S_1 = S\cup \{z_1\},S_2 = S \cup\{z_2\}$. Let $\pi_1$ and $\pi_2$ be the output functions of the pre-processing algorithm $\pi$ on $S_1$ and $S_2$ respectively. Our objective is to upper bound  the Rényi divergence between the output distribution of $\cA$ on the pre-processed datasets $\pi_1(S_1)$ and $\pi_2(S_2)$. 

We proceed by constructing a new dataset $\tilde{S}$ that consists of $\pi_2(S)$ and the point $\pi_1(z_1)$, as indicated in~\Cref{fig:thm1-proof-sketch}. The construction ensures that $\tilde{S}$ and $\pi_2(S_2)$ are neighboring datasets, and that the $L_{12}$ distance between $\tilde{S}$ and $\pi_2(S_2)$ is upper bounded by $\Delta_2\Delta_\infty$. We then apply the RDP property of $\cA$ to upper bound the divergence between $\cA(\tilde{S})$ and $\cA(\pi_2(S_2))$. Then, we employ the SRDP property of $\cA$ to upper bound the divergence between $\cA(\tilde{S})$ and $\cA(\pi_1(S_1))$. Finally, we establish the desired upper bound in~\Cref{eq:general-results} by combining the previous two divergences using the weak triangle inequality of Rényi divergence (\Cref{lem:triangle-inequality-renyi-divergence}). 
\end{proof}

While the privacy guarantee provided by~\Cref{thm:general-thm-training-only} is conditional over a dataset collection $\cL$, it can be extend to an unconditional privacy guarantee over all possible datasets using the Propose-Test-Release (PTR) framework \citep{dwork09ptr}. We show an example of the application of PTR and its guarantees in~\Cref{sec:ptr}.

\begin{figure}[t]
    \centering
    \includegraphics[width = 0.6\linewidth]{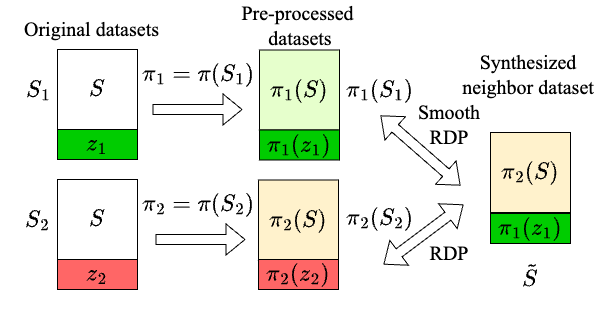}
    \caption{Illustration of the privacy analysis: For two neighboring datasets $S_1, S_2$, a pre-processing algorithm $\pi$ yields the pre-processed datasets $\pi_1(S_1)$ and $\pi_2(S_2)$ respectively. A synthetic dataset $\tilde{S}$ is constructed by combining the pre-processed datasets, ensuring that $\tilde{S}$ and $\pi_2(S_2)$ are neighboring datasets and that $\tilde{S}$ and $\pi_1(S_1)$ have bounded $L_{12}$ distance. }
    \label{fig:thm1-proof-sketch}
\end{figure}

\paragraph{Comparison with the group privacy or DDP analysis} A naive analysis on the privacy guarantee of $\cA\circ \pi$ using group privacy or DDP~\citep{Zhao17ddp,liu2016ddp} provides an upper bound that grows  polynomially\footnote{Linearly if we consider Approximate DP.} with \(\Delta_{\infty}\), which can be as large as $\text{poly}(n)$. In contrast, \Cref{thm:general-thm-training-only} implies a tighter bound on privacy for common DP mechanisms, especially when $\Delta_\infty\Delta_2 = o(n)$.\footnote{\Cref{thm:general-thm-training-only} is also applicable for other distance metrics, as long as Smooth RDP and the sensitivity of pre-processing are defined under comparable distance metric. } Next, we present various examples of pre-processing algorithms with small sensitivities $\Delta_\infty\Delta_2 = O(1)$. In~\Cref{fig:comparison-with-group-privacy}, we illustrate the improvement of the privacy analysis in~\Cref{thm:general-thm-training-only} over the conventional analysis via group privacy for numerous pre-processing algorithms.

\section{Privacy Guarantees of Common Pre-Processing Algorithms}
\label{sec:priv-guarantees-pre}
In this section, we use~\Cref{thm:general-thm-training-only} to provide overall privacy guarantees for several common pre-processing algorithms. First, in~\Cref{sec:sensitivity-preproc}, we define the pre-processing algorithms \(\pi\) and bound their \(L_2\) and \(L_{\infty}\) sensitivities. Then, in~\Cref{sec:overall-privacy} we provide the actual privacy guarantees for all combinations of these pre-processing algorithms and privacy mechanisms defined in~\Cref{sec:private-mechanisms}. We assume the instance space $\cX$ to be the Euclidean ball with radius $1$.

\subsection{Sensitivity of Common Pre-Processing Algorithms}
\label{sec:sensitivity-preproc}

\paragraph{Approximate deduplication} Many machine learning models, especially Large Language Models (LLMs), are trained on internet-sourced data, often containing many duplicates or near-duplicates. These duplicates can cause issues like memorization, bias reinforcement, and longer training times. To mitigate this, approximate deduplication algorithms are used in preprocessing, as discussed in \citet{2021gopher, 2023Llama}. We examine a variant of these algorithms termed \(\eta\)-approximate deduplication.

The concept of \(\eta\)-approximate deduplication involves defining a \textit{good cluster}. For a dataset \(S\) and a point \(x \in S\), consider \(B(x, \eta; S) = \{\tilde{x} \in S: \norm{\tilde{x}-x}_2 \leq \eta\}\), a ball of radius \(\eta\) around \(x\). This forms a \textit{good cluster} if \(B(x, \eta; S) = B(x, 3\eta; S)\). Essentially, this means that any point within a good cluster is at least \(2\eta\) distant from all other points outside the cluster in the dataset. We also define $B(S) = \{B_i = B(x_i, \eta; S)\}_{i = 1}^{m}$ as the set of all good clusters in a dataset $S$.


The \(\eta\)-approximate deduplication process, \(\pi_\eta^d\), identifies and retains only the center of each good cluster, removing all other points. Points are removed in reverse order of cluster size, prioritizing those with more duplicates.

\paragraph{Quantization}
Quantization, another pre-processing algorithm for data compression and error correction, is especially useful when the dataset contains measurement errors. We describe a quantization method similar to $\eta$-approximate deduplication, denoted as $\pi_\eta^q$: it identifies all \emph{good} clusters in the dataset, and replaces all points within each good cluster with the cluster's centroid in the reverse order of the size of the good clusters. The difference between de
duplication and quantization is that while quantization replaces near duplicates with a representative value, deduplication removes them entirely. We discuss the $L_2$ and $L_\infty$ sensitivity of deduplication and quantization in \Cref{prop:deduplication-sensitivities}.

\begin{restatable}{proposition}{deduplicationSensitivities}\label{prop:deduplication-sensitivities}
     For a dataset collection $\cL$, the $L_2$ and $L_\infty$ sensitivities\footnote{When the definition of neighboring dataset is refined to addition and deletion of a single data point, the $L_\infty$ sensitivity of both deduplication and quantization on a set $\cL$ can be reduced to $\max_{S\in \cL, B\in B(S)}|B|$. } of $\eta$-approximate deduplication $\pi_\eta^d$ and quantization $\pi_\eta^q$ are $\Delta_2(\cL, \pi_\eta^d) = 1$ and $ \Delta_2(\cL, \pi_\eta^q) = \eta$, and \[\Delta_\infty(\cL ,\pi_\eta^d)=\Delta_\infty(\cL, \pi_\eta^q)= \max_{S\in \cL}\max_{B \in B(S)}2\abs{B}.\]
\end{restatable}

While the $L_2$ sensitivity of deduplication is a constant $1$,  $\Delta_\infty$ is usually upper bounded by a small number. This is because in realistic datasets, the number of near duplicates is generally small for sufficiently small $\eta$. This leads to upper bounding the product $\Delta_2\Delta_\infty$ by a small number. For example, in text datasets, the fraction of near duplicates is typically smaller than 0.1, as demonstrated in Table 2 and 3 in~\citet{lee-etal-2022-deduplicating}. 

\paragraph{Model-based imputation}
Survey data, such as US census data, often contains missing values, resulting from the participants unable to provide certain information,  invalid responses, and changing questionnaire over time. Hence, it is crucial to process the missing values in these datasets with data imputation methods, to make optimal use of the available data for analysis while minimizing the introduction of bias into the results. 

We consider several imputation techniques which use the values of the dataset to impute the missing value. This can involve training a regressor to predict the missing feature based on the other feature, or simply imputing with dataset-wide statistics like mean, median or trimmed mean. For the sake of clarity, we only discuss mean imputation in the main text but provide the guarantees for other imputations in~\Cref{prop:imputation-sensitivity,tab:imputation-sensitivity} in~\Cref{app:imputation-tech}.



\begin{corollary}\label{prop:imputation-mean-sensitivity}
    For a dataset collection $\cL$ with maximum \(p\) missing values in any dataset, the $L_\infty$ sensitivity of mean imputation $\pi_{\mathrm{mean}}$ over $\cL$ is \(p\) and the $L_2$-sensitivity of $\pi_{\mathrm{mean}}$ is upper bounded by $\frac{2}{n-p}$.
\end{corollary}

\paragraph{Principal Component Analysis} 
Principal Component Analysis (PCA) is a prevalent pre-processing algorithm. It computes a transformation matrix $A_k^\top \in \reals^{k\times d}$ using the top $k$ eigenvalues of the dataset $S$'s covariance matrix. PCA serves two main purposes: dimension reduction and rank reduction. For dimension reduction, denoted as $\pi_{\mathrm{PCA-dim}}$, PCA projects data into a lower-dimensional space (typically for high-dimensional data visualization), using the pre-processing function $\pi_{S, \mathrm{dim}}(x) = A_k^\top x$. For rank reduction, represented as $\pi_{\mathrm{PCA-rank}}$, PCA leverages the low-rankness of the dataset with the function $\pi_{S, \mathrm{rank}}(x) = A_kA_k^\top x$.

The primary difference between these two PCA applications is in the output dimensionality. Dimension reduction yields data of dimension $k$, while rank reduction maintains the original dimension $d$, but with a low-rank covariance matrix of rank $k$. We detail the $L_\infty$ and $L_2$ sensitivity for both PCA variants in \Cref{prop:pca-sensitivity}.

\begin{restatable}{proposition}{PCASensitivity}
    \label{prop:pca-sensitivity}
    For a dataset collection $\cL$, the $L_\infty$ sensitivity of $\pi_{\mathrm{PCA-dim}}$ and $\pi_{\mathrm{PCA-rank}}$ is the size of the datasets in $\cL$, \ie \(~n\).~
    The $L_2$-sensitivity of $\pi_{\mathrm{PCA-dim}}$ and $\pi_{\mathrm{PCA-rank}}$ is bounded by $2\Delta_2$ and $\Delta_2$ respectively, where \[\Delta_2 = \frac{4(3n + 2)}{n(n-1)\min\{\delta_{\min}^k(\cL), \delta_{\min}^1(\cL)\}},\]
    where  $\delta_{\min}^k(\cL) = \min_{S\in \cL} \lambda_k(S) - \lambda_{k+1}(S)$ is the minimum gap between the $k^{\text{th}}$ and $(k+1)^{\text{th}}$ eigenvalue over any covariance matrix of $S \in \cL$. 
\end{restatable}

\paragraph{Standard Scaling} Scaling is one of the most common pre-processing methods. In~\Cref{prop:sensitivity-scaling}, we provide the sensitivity results of standard scaling which scales each feature to have mean 0 and standard deviation 1, and min max scaling, which scales each feature to the interval between 0 and 1.

\begin{restatable}{proposition}{StandardScalingSensitivity}
\label{prop:sensitivity-scaling}
    For a dataset collection $\cL$, the $L_\infty$ sensitivity of standard scaling and min max scaling is the size of the datasets in $\cL$, i.e. $n$. The $L_2$ sensitivity of standard scaling is \[\Delta_2 = \frac{2}{\sigma_{\min}^3 n} + \frac{2}{n\sigma_{\min}},\]
    where $\sigma_{\min}$ is the minimum standard deviation over datasets in $\cL$.
\end{restatable}

\subsection{Privacy Analysis for Pre-Processing Algorithms}
\label{sec:overall-privacy}

\begin{table*}[t]
\caption{Overall privacy guarantees $\widehat{\varepsilon}$ of pre-processed DP pipelines with $\alpha \geq 11$. Let $p$ represent the $L_\infty$ sensitivity of deduplication, quantization and mean imputation. We also assume the size of the dataset $n \geq 101$ for PCA, and the Lipschitz and smoothness parameters, along with global sensitivity, are set to $1$. We omit the privacy guarantees of deduplication due to space constraint. However, we note that the privacy guarantees for deduplication are the same as those for quantization with $\eta = 1$. See~\Cref{app:sec4-proofs} for details. }
\centering\small
\begin{tabular}{C{0.1\linewidth}C{0.2\linewidth}C{0.2\linewidth}C{0.2\linewidth}C{0.2\linewidth}}
\hline
\toprule
 & Quantization & Mean imputation & PCA & Standard Scaling\\
\midrule
$\cM_G$ & ${1.05\alpha\varepsilon^2}\br{1+\eta^2p^2}$& ${1.05\alpha\varepsilon^2}\br{1+\frac{4p^2}{(n-p)^2}}$& ${1.05\alpha\varepsilon^2}\br{1+\frac{12.2^2}{(\delta_{\min}^k)^2}}$&${1.05\alpha\varepsilon^2}\br{1+\frac{4}{\sigma_{\min}^3}}$\\
$\cA_{\mathrm{GD}}$&${1.05\alpha\varepsilon^2}\br{4+\eta^2p^2}$&${4.2\alpha\varepsilon^2}\br{1+\frac{p^2}{(n-p)^2}}$&${1.05\alpha\varepsilon^2}\br{4+\frac{12.2^2}{(\delta_{\min}^k)^2}}$&${4.2\alpha\varepsilon^2}\br{1+\frac{1}{\sigma_{\min}^3}}$\\
$\cM_L/\cM_E$ &$\varepsilon\br{1 + \eta p} $ & $\varepsilon\br{1 + \frac{2p}{n-p}}$  &$\varepsilon\br{1+\frac{12.2}{\delta_{\min}^k}}$&$\varepsilon\br{1+\frac{4}{\sigma_{\min}^3}}$   \\
$\cA_{\mathrm{SGD-samp}}$&--&--&${1.05\alpha\varepsilon^2}\br{2\alpha + \frac{12.2^2}{(\delta_{\min}^k)^2}}$&${2.1\alpha\varepsilon^2}\br{\alpha + \frac{8}{\sigma_{\min}^6}}$\\

$\cA_{\mathrm{SGD-iter}}$&${1.1\alpha\varepsilon^2}\br{1+\frac{\frac{\eta^2p^2n}{\log n}}{\frac{n-p}{\log n-p}}}$&${1.1\alpha\varepsilon^2}\br{1+\frac{4p^2\frac{n}{\log n }}{\frac{(n-p)^3}{\log n-p}}}$&--&--\\
\bottomrule
\end{tabular}
\label{tab:comparison}
\end{table*}

\begin{figure*}[t]
\centering
\begin{subfigure}{0.32\textwidth}
    \includegraphics[width=\textwidth]{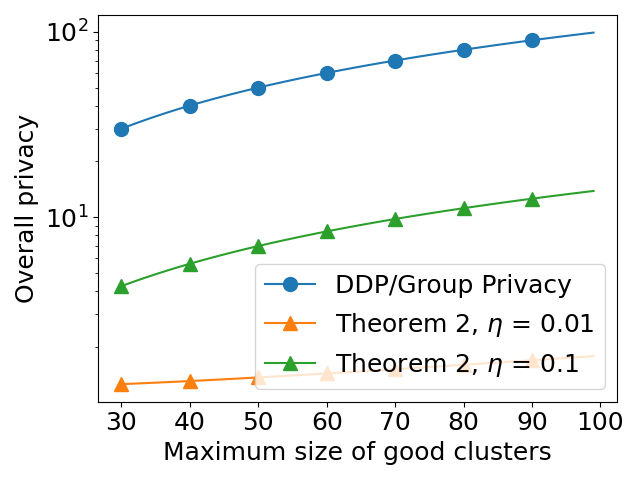}
    \caption{Quantization}
\end{subfigure}
\hfill
\begin{subfigure}{0.32\textwidth}
    \includegraphics[width=\textwidth]{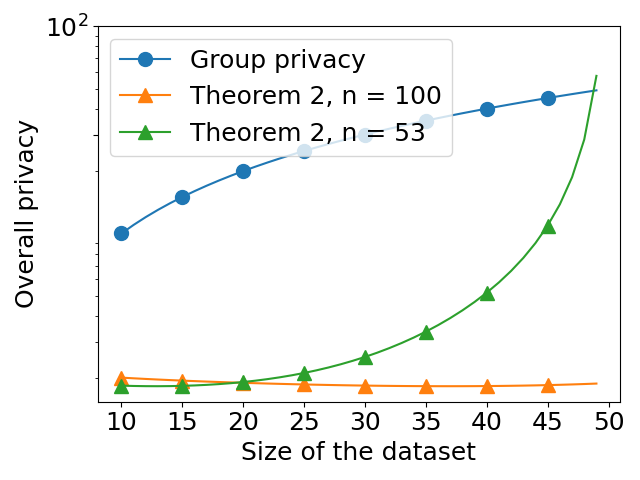}
    \caption{Mean Imputation}
    \label{subfig:mean-imputation-bound}
\end{subfigure}
\hfill
\begin{subfigure}{0.32\textwidth}
    \includegraphics[width=\textwidth]{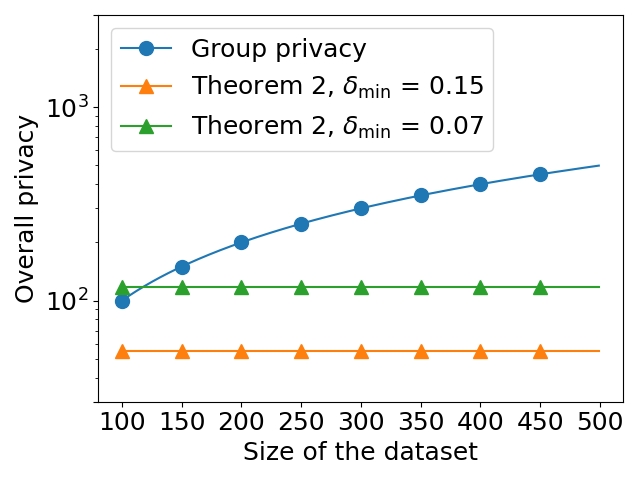}
    \caption{PCA for rank reduction}
\end{subfigure}
        
\caption{Visualization of the overall privacy of pre-processed Gaussian mechanism analysed with group privacy and our bound from Theorem 2. Here, $\eta$ is the distance threshold of the quantization algorithm, $n$ is the size of the possible datasets, and $\delta_{\min}$ is the minimum gap between the $k^{th}$ and the $k+1^{th}$ eigenvalue of all possible datasets. }
\label{fig:comparison-with-group-privacy}
\end{figure*}

After establishing the necessary elements of our analysis, including the sensitivities of pre-processing algorithms (\Cref{sec:pre-processing}) and the SRDP parameters of private mechanisms (\Cref{tab:sdp-parameters}), we are ready to present the exact overall privacy guarantees for specific pre-processed DP pipelines. In~\Cref{tab:comparison}, we present the privacy guarantees for various combinations of pre-processing methods and private mechanisms. 

\begin{thm}[Informal]\label{thm:overall-privacy-table2}
Pre-processed DP pipelines comprised of all combinations of private mechanisms in~\Cref{sec:private-mechanisms} and pre-processing algorithms in~\Cref{sec:sensitivity-preproc}  are $(\alpha, \widehat{\varepsilon})$-RDP where $\widehat{\varepsilon}$  is specified in~\Cref{tab:comparison}. 
\end{thm}

\Cref{tab:comparison} shows that the pre-processing methods discussed in~\Cref{sec:sensitivity-preproc} typically lead to a minimal, constant-order increase in the privacy cost for common DP mechanisms, depending on the datasets in $\cL$. Specifically, deduplication results in a constant increase in the privacy cost when the datasets in $\cL$ do not contain large clusters, \ie~$\Delta_\infty = o(n)$, whereas quantization can handle larger clusters of size $\Delta_\infty = o(\nicefrac{n}{\eta})$. The privacy cost of mean imputation remains constant for datasets with few missing values ($p = o(n)$). Standard scaling leads to constant amplification in privacy cost when all features in dataset collection $\cL$ have non-vanishing standard deviation. For PCA, we only present the results for rank reduction and the results for dimension reduction follows similarly. The additional privacy cost remains small for low rank datasets with $\delta_{\min}$ bounded away from 0. 

Moreover,~\Cref{fig:comparison-with-group-privacy} demonstrates a comparison between our SRDP-based analysis (\Cref{thm:general-thm-training-only}) and the naive group privacy/DDP analysis for pre-processed Gaussian mechanism. This comparison illustrates the advantage of our SRDP-based analysis over group privacy/DDP and how the privacy guarantees vary with properties of the dataset collections $\cL$. For quantization, the SRDP-based privacy analysis provides significantly stronger privacy guarantees for smaller \(\eta\). For mean imputation, the difference between SRDP-based analysis and group pirvacy analysis is large for smaller number of missing values. However, in the extreme scenario where over 90\% of the data points are missing, group privacy achieves a similar guarantee as our analysis with SRDP. In the case of PCA for rank reduction, the privacy parameter obtained with SRDP-based privacy analysis remains constant, while it increases for group privacy analysis with the size of the dataset.

\section{Unconditional Privacy Guarantees in Practice}
\label{sec:practical-implications}
The previous section provides a comprehensive privacy analysis for pre-processed DP pipelines where the resulting privacy guarantee depends on properties of the dataset collection. In this section, we first illustrate how conditional privacy analysis can become ineffective for pathological datasets~(\Cref{sec:limitation-conditional-privacy}). Then, we introduce a PTR-inspired framework to address this and establish unconditional privacy guarantees in~\Cref{sec:ptr}. Using PCA for rank reduction as an example, we provide convergence guarantee for excess empirical loss for generalized linear models and validate our results with synthetic experiments. 

\subsection{Limitations of conditional privacy guarantee}\label{sec:limitation-conditional-privacy}

The previous analysis in~\Cref{sec:pre-processing,sec:priv-guarantees-pre} rely on the chosen dataset collection $\cL$. Typically, for dataset collections \(\cL\) that are well-structured, non-private pre-processing leads to only a slight decline in the privacy parameters as discussed in~\Cref{sec:overall-privacy}. However, this degradation can become significant for datasets that exhibit pathological characteristics. In the following, we present several examples where our guarantees in~\Cref{thm:general-thm-training-only} become vacuous due to the pathological nature of the dataset. 

    \paragraph{Imputation} If the number of missing points is comparable to the size of the dataset i.e. \(p\approx n\), the privacy guarantee in~\Cref{thm:general-thm-training-only} deteriorates to the level of those obtained from group privacy or DDP. This deterioration is also reflected in~\Cref{subfig:mean-imputation-bound}.

    \paragraph{PCA} When performing PCA with a reduction to rank $k$, the privacy guarantee can worse and scale with $n$ for very small $\delta_{\min}^k$, in particular when $\delta_{\min}^k=\bigO{\nicefrac{1}{n}}$. This situation can occur naturally when the data is high rank or when \(k\) is chosen ``incorrectly''.
    
    \paragraph{Deduplication} Consider a dataset $S = \{z_1, ..., z_n\}$, where each pair of distinct data points are at least $\eta$ distance apart, yet each point are no more than $\eta$ distance from a specific reference point $x$, \ie $\norm{z_i - z_j}_2\geq \eta$, $\norm{z_i - x}_2\leq \eta$, for $i\neq j\in [n]$.Under these conditions, applying deduplication on $S$ leaves the dataset unchanged as all points are uniformly \(\eta\) distance from each other. However, if \(x\) is added to \(S\), then deduplication would eliminate all points except \(x\). As a result, the $L_\infty$ sensitivity of deduplication becomes $n$, resulting in the same privacy analysis as group privacy. Interestingly, a similar example was leveraged by~\citet{debenedetti2023sidechannels} in their side-channel attack.


It is important to recognise that datasets exhibiting pathologucal characteristics, like those above, are usually not of practical interest in data analysis. For instance, applying mean imputation to a dataset where the number of missing points is nearly equal to the size of the dataset or implementing approximate deduplication that results in the removal of nearly all data are not considered sensible practice. A possible solution to this problem is to potentially refuse to provide an output when the input dataset is deemed pathological, provided that the decision to refuse is made in a manner that preserves privacy. This approach is adopted in the following section, where we employ the Propose-Test-Release framework proposed by~\citet{dwork09ptr} to establish unconditional privacy guarantees over all possible datasets, at the expense of accuracy on datasets that are ``pathological''. 

\subsection{Unconditional privacy guarantees via PTR}\label{sec:ptr}

We present~\Cref{alg:ptr-imputation} that applies the PTR procedure to combine the non-private PCA for rank reduction with DP-GD. It exhibits unconditional privacy guarantee over all possible datasets (\Cref{thm:ptr-privacy-guarantee}). While this technique is specifically described for combining PCA with DP-GD, it's worth noting that the same approach can be applied to other pairings of DP mechanisms and pre-processing algorithms, though we do not explicitly detail each combination since their implementations follow the same basic principles.

\begin{algorithm}
    \caption{PTR for $\pi_{\mathrm{PCA-rank}}$ on DP-GD}
    \textbf{Input:} Dataset $S$, estimated lower bound $\beta$ of $\delta_k(S)$, privacy parameters $\varepsilon, \delta$, Lipschitzness  $L$
    \begin{algorithmic}[1]
        \STATE Set
            $\Gamma = \min_{{S':\delta_{k}(S')< \beta}}d_H(S, S') + Lap\br{\frac{1}{\varepsilon}}$\qquad\qquad\qquad\qquad\qquad\qquad\qquad \(\delta_{k}(S') = \lambda_k\br{S'} - \lambda_{k+1}\br{S'}\)
        \IF{$\Gamma \leq {(\log\frac{2}{\delta})}/{\varepsilon}$}
        \STATE Return $\perp$ 
        \ELSE 
        \STATE Compute the pre-processing function $\pi_{\mathrm{PCA-rank}}^S$ using the dataset $S$ and set $\sigma = \frac{2L\sqrt{T}}{\varepsilon n}$.
        \STATE Return $\cA_{\mathrm{DP-GD}}(\pi_{\mathrm{PCA-rank}}^S(S))$ with parameter $\sigma$. 
        \ENDIF
    \end{algorithmic}
    \label{alg:ptr-imputation}
\end{algorithm}

\begin{restatable}{thm}{PTRGuarantee}\label{thm:ptr-privacy-guarantee}
    For any $L$-Lipschitz and $\mu$-smooth loss function $\ell$, $\varepsilon > 0$ and $\delta \leq\exp\br{-1.05\varepsilon^2\br{1 + \frac{12.2^2\mu^2}{L^2\beta^2}}}$, \Cref{alg:ptr-imputation} with privacy parameters $\varepsilon, \delta$, and estimated lower bound $\beta$ is $\br{\widehat{\varepsilon} +\varepsilon,\delta}$-DP on a dataset of size $n\geq 101$, where $\widehat{\varepsilon} = 3\varepsilon\sqrt{1.05\br{1 + \frac{12.2^2\mu^2}{L^2\beta^2}}\log\frac{1}{\delta}}$. 
\end{restatable}
The RDP parameter on PCA with DP-GD in~\Cref{tab:comparison} with the same parameter $\sigma$ can be converted to $\br{\widehat{\varepsilon},\delta}$- DP if it holds that for all \(S\in\cL\), $\delta_{k}\br{S} \geq \beta$ . In comparison, the guarantee in~\Cref{thm:ptr-privacy-guarantee} is marginally worse by an additive $\varepsilon$, but it remains applicable even when $\delta_{k}\br{S}< \beta$ for some \(S\in\cL\). However, when $\delta_{k}\br{S}< \beta$ for some \(S\in\cL\), the bound in~\Cref{tab:comparison} degrades to the same as obtained via group privacy/DDP.


For generalized linear models, we present an high probability upper bound on the excess empirical loss of~\Cref{alg:ptr-imputation} conditional on the properties of the private dataset. Given a dataset $S = \{(x_i, y_i)\}_{i = 1}^n$ and a loss function $\ell$, let $\hat{\ell}_S(\theta) = \frac{1}{n}\sum_{i = 1}^n \ell(\langle\theta, x_i\rangle, y_i)$ denote the empirical loss of a generalized linear model on the dataset $S$ and let $\theta^\star = \text{arg}\min_{\theta\in \reals^d} \hat{\ell}_S(\theta)$. For simplicity, we assume \(S\) is centered and \(\ell\) is L-Lipschitz. 

\begin{restatable}{proposition}{PCAaccuracy}\label{prop:pca-accuracy}
For $\delta, \varepsilon$ defined in~\Cref{alg:ptr-imputation}, $n\geq 101$, and any \(\beta \leq\delta_k\br{S}\), with probability at least $1-\xi$, ~\Cref{alg:ptr-imputation} outputs $\hat{\theta}$ such that the excess empirical risk 
     \begin{equation}
         \label{eq:ptr-pca-acc-guarantee}
         \bE\bs{\hat{\ell}_S(\hat{\theta})} - \hat{\ell}_S(\theta^\star)= \bigO{\frac{L\sqrt{1 + k\log \nicefrac{1}{\delta}}}{\varepsilon n}}+ 2L\Lambda, 
     \end{equation}
   \small  where $\xi = \frac{1}{2\delta}\exp\br{- \frac{(\delta_{k}(S) 
     - \beta)n\varepsilon}{12.2} }, \Lambda = \sum_{i = k+1}^d\lambda_{i}(S)$  where the high probability is over the randomness in Step 1 and the expectation is over the randomness of Step 6.
\end{restatable}


The results follow a similar analysis as \cite{song21Evading}, which shows that the convergence bound of DP-GD scales with the rank of the dataset. However, when the dataset remains full rank but the first $k$ eigenvalues dominates the rest, e.g. $\Lambda = O(\sqrt{k})$, the convergence bound following \cite{song21Evading} is $O(\sqrt{d})$. In contrast,~\Cref{prop:pca-accuracy} leads to a dimension-independent convergence bound of order $O(\sqrt{k})$, with a slight degradation in the privacy guarantee. Here, \(\beta\) introduces the trade-off between privacy and utility. Large \(\beta\) leads to tighter privacy guarantee~(small effective \(\varepsilon\) in~\Cref{thm:ptr-privacy-guarantee}) at the risk of worse utility~(large \(\xi\) in~\Cref{prop:pca-accuracy}). 

\paragraph{Experimental results} We conducted experiments on a synthetic approximately low rank dataset to corroborate our results in~\Cref{prop:pca-accuracy}, and summarised the results in~\Cref{fig:training-acc-pca}. We generate a 2-class low rank dataset consisting of 1000 data points with dimension 6000 and approximately rank 50. The synthetic dataset has positive yet small eigenvalues for the $k^{\text{th}}$ eigenvectors for $k\geq 50$, ensuring $2L\Lambda$ in~\Cref{eq:ptr-pca-acc-guarantee} is small but positive. 

For each overall privacy parameter $\varepsilon$, we evaluated the excess empirical risk of DP-SGD with: a) non-private PCA with an adjusted privacy parameter from~\Cref{tab:comparison}, b) DP-PCA, with half of the privacy budget to allocated to DP-PCA, and c) no pre-processing. 

\Cref{fig:training-acc-pca} shows that pre-processed DP pipeline outperforms logistic regression without pre-processing. This is because the dataset's approximate low rankness, indicated by $2L\Lambda > 0$, prevents logistic regression from leveraging the dimension-independent optimization guarantees using the original high dimensional datasets~\citep{song21Evading}. However, non-private PCA, while incurring a minor constant-order privacy cost, effectively utilizes the data's low rankness and offers significant benefits, especially at smaller $\varepsilon$ values where optimization error is dominated by the first term in~\Cref{eq:ptr-pca-acc-guarantee}.~\Cref{fig:training-acc-pca} also demonstrates that DP-PCA exhibits the worst performance among the three methods, beccause the inherent error of order $\Omega(d)$ of DP-PCA dominates given a high dimensional dataset with $d \ll n$\citep{Liu22dp-pca}. 

In~\Cref{app:experiments}, we provide the details on experimental setups and a discussion on the practical modification of~\Cref{alg:ptr-imputation} with clipping.

\begin{figure}
\centering
\includegraphics[width=0.4\textwidth]{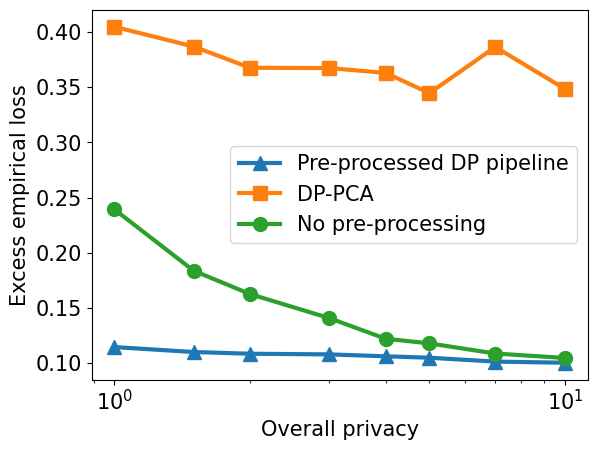}
\caption{Comparison of excess empirical loss of private logistic regression: for each level of overall privacy $\varepsilon$, pre-processed DP pipeline consistently outperforms other methods.}
\label{fig:training-acc-pca}
\end{figure}

        

\section{Conclusion}
In this paper, we investigate the often-neglected impact of pre-processing algorithms in private ML pipelines. We propose a framework to assess the additional privacy cost from non-private pre-processing steps using two new technical notions: Smooth RDP and sensitivities of pre-processing algorithms. Finally, we propose a PTR-based procedure to relax some of the necessary assumptions in our framework and make it practically usable. Several interesting directions of future work remain unexplored, including handling more complex pre-processing algorithms, such as pre-trained deep neural feature extractors and private algorithms like private data synthesis.

\section*{Acknowledgements}
We thank Francesco Pinto for the initial discussion of the project. We also thank Alexandru \c{T}ifea, Jiduan
Wu and Omri Ben-Dov for their useful feedback. Additionally, we extend our gratitude to the reviewers of ICML for their careful reviews and insightful suggestions.


\bibliography{arxiv_main}
\bibliographystyle{icml2024}

\newpage
\appendix
\onecolumn

We present the detailed proof of our results and some additional findings in the appendix. To distinguish the lemmas from existing literature and new lemmas used in our proofs, we adopt alphabetical ordering for established results and numerical ordering for our contributions. 

\section{Proofs regarding Smooth RDP (\Cref{sec:srdp})}

In this section, we present the proofs regarding SRDP properties. In \Cref{app:sec2-proofs}, we introduce SRDP satisfies some desired properties, and present their proofs.~\Cref{app:srdp-table} provides the exact SRDP parameters of common DP mechanisms.

\subsection{Properties of SRDP and their proofs}

While conditional SRDP over a set $\cL$ can be considered as a special case of Pufferfish R\'enyi Privacy~\citep{kifer14pufferfish}, we show that it satisfies desirable properties, such as sequential composition and privacy amplification by subsampling, which are not satisfied by the more general Pufferfish Renyi Privacy \citep{Pierquin2023renyi}. 

\paragraph{Properties of SRDP} Similar to RDP, SRDP satisfies (sequential) composition and closure under post-processing. 

\begin{restatable}{lem}{propertiesSRDP}\label{lem:properties-srdp}
   Let $\cL$ be any dataset collection. Then, the following holds.
   \begin{itemize}[leftmargin= 0.8em]
       \item\textbf{Composition} Let $\alpha > 0, \varepsilon_i: \reals\times \reals \to \reals, \forall i\in [k]$. For any $i\in [k]$, if the randomized algorithms $\cA_i$ is $(\alpha, \varepsilon_i(\alpha, \sdpL))$-SRDP over $\cL$, then the composition $(\cA_1, \ldots, \cA_k)$ is $(\alpha, \sum_{i = 1}^k \varepsilon_i(\alpha, \sdpL))$-SRDP over $\cL$. 
       \item\textbf{Post-processing} Let $\alpha > 0, \varepsilon: \reals \times \reals \to \reals$, and $f$ be an arbitrary algorithm. For any $\sdpL > 0$, if $\cA$ is $(\alpha, \varepsilon(\alpha, \sdpL))$-SRDP over $\cL$, then $f\circ\cA$ is $(\alpha, \varepsilon(\alpha, \sdpL))$-SRDP over $\cL$. 
   \end{itemize} 
\end{restatable}

SRDP also satisfies a form of Privacy amplification by subsampling. We state the weaker version without additional assumptions in~\Cref{app:sec2-proofs} and use a stronger version for $\cA_{\mathrm{SGD-iter}}$ in~\Cref{thm:sdp-parameters} with some light assumptions.

\label{app:sec2-proofs}
\propertiesSRDP*
\begin{proof}
    \textbf{Composition: } For $i\in [k]$ and any two datasets $S, S'\in \cL$ with $d_{12}(S, S')\leq \sdpL$, let $\nu_i$ and $\nu_i'$ be the distribution of $\cA_i(S)$ and $\cA_i(S')$ respectively. Then, using to the independence between $\cA_i(S)$ and $\cA_{i+1}(S)$,  denote the joint distribution of $\cA_i(S)$ and $\cA_{i+1}(S)$ as $\nu_i\times \nu_{i+1}$ and the joint distribution of $\cA_i(S')$ and $\cA_{i+1}(S')$ as $\nu_i'\times \nu_{i+1}'$. 

    Then, we prove the composition property of by upper bounding $\alphaRDP{\cA_1(S), \ldots, \cA_k(S)}{\cA_1(S'), \ldots, \cA_k(S')}$ for $S, S'\in \cL$ with $d_{12}(S, S') \leq \sdpL$. To achieve this, we employ the following lemma on additivity of Rényi divergence (\Cref{lem:additivity-renyi-divergence}). 
    
    \begin{lemL}[Additivity of Renyi divergence \citep{Erven12RenyiDivergence}]\label{lem:additivity-renyi-divergence}
        For $\alpha > 1$ and distributions $\nu_1, \nu_2, \nu_1', \nu_2'$, \[\alphaRDP{\nu_1\times \nu_2}{\nu_1'\times \nu_2'} = \alphaRDP{\nu_1}{\nu_1'} + \alphaRDP{\nu_2}{\nu_2'}. \]
    \end{lemL}
    \noindent Applying~\Cref{lem:additivity-renyi-divergence} at step (a),
    \begin{equation}
    \begin{aligned}
    \alphaRDP{\cA_1(S), \ldots, \cA_k(S)}{\cA_1(S'), \ldots, \cA_k(S')} &= D_\alpha(\nu_1\times \ldots \times \nu_{k}||\nu_1'\times \ldots \times \nu_{k}') \\
            &\overset{(a)}{\leq} \alphaRDP{\nu_1}{\nu_1'} + \alphaRDP{\nu_2\times \ldots \times \nu_{k}}{\nu_2'\times \ldots \times \nu_{k}'}\\
            &\overset{(b)} {=}\varepsilon_1(\alpha, \sdpL) + \alphaRDP{\nu_2\times \ldots \times \nu_{k}}{\nu_2'\times \ldots \times \nu_{k}'}\\
            &\overset{(c)}{=}  \varepsilon_1(\alpha, \sdpL) 
             + \varepsilon_2(\alpha, \sdpL) 
             + \ldots \\
            &\leq \sum_{i = 1}^k\varepsilon_i(\alpha, \sdpL)
    \end{aligned}
    \end{equation}
    where step (b) follows from the fact that $\cA_1$ is SRDP with parameter $\varepsilon_1(\alpha, \sdpL)$, and step (c) is obtained by decomposing $D_\alpha(\nu_2\times \ldots \times \nu_{k}||\nu_2'\times \ldots \times \nu_{i+1}')$ in the similar manner as step (a) and applying~\Cref{lem:additivity-renyi-divergence} iteratively. This completes the proof for composition of SRDP. 

    \textbf{Post-processing: }For any $S, S'\in \cL$ with $d_{12}(S, S')\leq \sdpL$, let $X = \cA(S)$, $X' = \cA(S')$ be random variables with probability distribution $\nu_X, \nu_X'$, and $Y = f(\cA(S)), Y' = f(\cA(S'))$ be random variables with distributions $\nu_Y$ and $\nu_Y'$. For any value $x$, we denote the conditional distribution of $Y = f(x)$ and $Y' = f(x)$ by $\nu_{Y|x}$ and $\nu_{Y|x}'$ respectively. We note that $\nu_{Y|x} = \nu_{Y|x}'$ by definition. 

    We will upper bound the Rényi divergence between $f(\cA(S))$ and $f(\cA(S'))$, 
    \begin{equation}
        \begin{aligned}
            \alphaRDP{f(\cA(S))}{f(\cA(S'))} &= \frac{1}{\alpha-1} \log \bE_{y\sim \nu_Y'}\bs{\frac{\nu_Y(y)}{\nu_Y'(y)}}^\alpha\\
            &= \frac{1}{\alpha - 1}\log \bE_{y\sim \nu_Y'}\bs{\bE_{x\sim \nu_X'}\br{\frac{\nu_{Y|x}(y)\nu_X(x)}{\nu_{Y|x}'(y)\nu_X'(x)}}}^\alpha\\
            &\overset{(a)}{\leq} \frac{1}{\alpha - 1}\log \bE_{y\sim \nu_Y'}\bs{\bE_{x\sim \nu_X'}\br{\frac{\nu_{Y|x}(y)\nu_X(x)}{\nu_{Y|x}'(y)\nu_X'(x)}}^\alpha}\\
            &\overset{(b)}{=}\frac{1}{\alpha - 1}\log \bE_{x\sim \nu_X'}\bs{\frac{\nu_X(x)}{\nu_X'(x)}}^\alpha = \alphaRDP{\cA(S)}{\cA(S')}
        \end{aligned}
    \end{equation}
    where $(a)$ follows from Jensen's inequality and the convexity of $h(x) = x^\alpha$ for $\alpha \geq 2$ and $x \geq 0$, and $(b)$ follows from the fact that $\nu_{Y|x} = \nu_{Y|x}'$ for fixed $x$. This completes the proof. 
\end{proof}

We present a general version of privacy amplification by subsampling for SRDP in~\Cref{lem:full-pas-srdp}. We consider a subsampling mechanism that uniformly selects $B$ elements from a dataset with replacement. \Cref{lem:full-pas-srdp} implies that the SRDP parameter of a SRDP algorithm decreases by $O(1/\sqrt{n})$ when $B = 1$ for datasets of size $n$. 

\begin{lem}[Privacy amplification by subsampling of SRDP] \label{lem:full-pas-srdp}
For $B \geq 1$, let $\pi: \cX^n\rightarrow \cX^B$ be the subsampling mechanism that samples $B$ elements from the dataset $S$ of size $n$ uniformly at random. Let $\cA$ be a randomized algorithm that is $(\alpha, \varepsilon(\alpha, \sdpL))$-SRDP. For $\alpha, \sdpL$ such that $\varepsilon(\alpha, \sdpL)\leq \frac{1}{\alpha-1}$ and for any integer $\frac{1}{B}\leq k\leq \frac{n-1}{B}$, $\cA\circ \pi$ satisfies $(\alpha, \varepsilon'(\alpha, \sdpL))$-SRDP, where 
\[\varepsilon'(\alpha, \sdpL) = 2\br{1-\frac{kB-1}{n}}^B\varepsilon\br{\alpha, \frac{\sdpL}{k}} + 2\br{1-\br{1-\frac{kB-1}{n}}^B}\varepsilon\br{\alpha, \sdpL}. \]
\end{lem}

\begin{proof}
    Let $S, S'\subset \cL$ with $d_{12}(S, S')\leq \sdpL$. First, order the points in $S, S'$ as $S = \{z_i\}_{i = 1}^n$ and $S' = \{z_i'\}_{i = 1}^n$ such that $\sum_{i = 1}^n\norm{z_i - z_i'}_2\leq \sdpL$. 
    
    Let $m$ be a fixed integer. Let $\cI = \{i_1, i_2, \ldots \}$ be the set of indices such that for any $j\in \cI$, $\norm{z_{j} - z_{j}'}\geq \frac{\sdpL}{m}$ holds. For $1\leq m \leq n-1$, if $d_{12}(S, S')\leq \sdpL$, there can be at most $m-1$ indices in $\cI$. 
    
    Thus, for an index $i$ sampled from $[n]$ uniformly at random, we have \begin{equation}
        \label{eq:pas1}
        \bP_{i\sim \text{Unif}([n])}\bs{\norm{z_i - z_i'}\leq \frac{\sdpL}{m}} = 1- \bP_{i\sim \text{Unif}([n])}\bs{\norm{z_i - z_i'}\geq \frac{\sdpL}{m}} \geq 1 - \frac{k-1}{n}.
    \end{equation}
    Let $J = \{i_1, ..., i_B\}$ be a set where each index $i_j$ is sampled independently and uniformly from $[n]$. For any integer $1\leq m \leq n-1$, we compute the probability that $d_{12}(S_J, S'_J)\leq \frac{\sdpL}{k}$. 
    \begin{equation}\label{eq:pas2}
        \bP_{J\sim \text{Unif}([n])^B}\bs{\sum_{i\in J}\norm{z_i - z_i'}_2\leq \frac{B\sdpL}{m}}\overset{(a)}{\geq} \bP_{J\sim \text{Unif}([n])^B}\bs{\forall i\in J, \norm{z_i - z_i'}_2\leq \frac{\sdpL}{m}} \overset{(b)}{\geq} \br{1-\frac{m-1}{n}}^B
    \end{equation}
    where step (a) follows as $\forall i\in J, \norm{z_i - z_i'}_2\leq \frac{\sdpL}{m}$ implies $\sum_{i\in J}\norm{z_i - z_i'}_2\leq \frac{B\sdpL}{m}$, step (b) follows by~\Cref{eq:pas1} and the independence of each index in $J$. 

    Replacing $m = kB$, we have for $1\leq kB\leq n-1$, \ie for $\frac{1}{B}\leq k \leq \frac{n-1}{B}$,
    \begin{equation}\label{eq:pas3}
        \bP_{J\sim \text{Unif}([n])^B}\bs{\sum_{i\in J}\norm{z_i - z_i'}_2\leq \frac{\sdpL}{k}}\geq \br{1-\frac{kB-1}{n}}^B
    \end{equation}

    Next, we apply the weak convexity of Rényi divergence (\Cref{lem:weak-convexity-renyi-divergence}) to get the desired bound.  
    \begin{lemL}[Weak convexity of Renyi divergence, Lemma 25 in~\citet{feldman18iteration}]\label{lem:weak-convexity-renyi-divergence}
        Let $\mu_1, \mu_2$ and $\upsilon_1, \upsilon_2$ be probability distributions over the same domain such that $\alphaRDP{\mu_i}{\upsilon_i}\leq \frac{c}{\alpha -1}$ for $i\in \{1, 2\}$ and $c\in (0, 1]$. For $\beta \in (0, 1)$, \[\alphaRDP{\beta\mu_1 + (1-\beta)\mu_2}{\beta\upsilon_1 + (1-\beta)\upsilon_2}\leq (1+c)\bs{\beta \alphaRDP{\mu_1}{\upsilon_1} + (1-\beta)\alphaRDP{\mu_2}{\upsilon_2}}. \]
    \end{lemL}
    Let $\cE$ denote the event that the $L_{12}$ distance between $S_J$ and $S'_J$ is smaller than $\frac{\sdpL}{k}$, \ie $\sum_{i\in J}\norm{z_i - z_i'}_2\leq \frac{\sdpL}{k}$, for $J\sim \text{Unif}([n])^B$. Then, \Cref{eq:pas3} implies that the probability of $\cE$ is at least $\br{1-\frac{kB - 1}{n}}^B$. 
    
    Let $\mu_1, \upsilon_1$ be the distribution of $\cA(S_J)$ and $\cA(S'_J)$ conditional on $\cE$ occurring, and let $\mu_2, \upsilon_2$ be the distribution of $\cA(S_J)$ and $\cA(S'_J)$ conditional on $\cE$ does not occur. By our assumption on $\sdpL, \alpha$, $\varepsilon(\alpha, \sdpL)\leq \frac{1}{\alpha - 1}$ and $c= 1$. Applying~\Cref{lem:weak-convexity-renyi-divergence} with $c = 1$ in step (a), 

    \begin{equation*}
    \begin{aligned}
        \alphaRDP{\cA(S_J)}{\cA(S'_J)} &= \alphaRDP{\bP\bs{\cE}\cA({S_J}{|\cE}) + \bP\bs{\neg\cE}\cA(S_J{|\neg\cE})}{\bP\bs{\cE}\cA(S'_J{|\cE}) + \bP\bs{\neg\cE}\cA(S'_J{|\neg\cE})}\\
        &\overset{(a)}{\leq} 2\bP\bs{\cE}\varepsilon\br{\alpha, \frac{\sdpL}{k}} + 2\bP\bs{\neg \cE}\varepsilon\br{\alpha, \sdpL}\\
        &\overset{(b)}{\leq} 2 \br{1 - \frac{kB - 1}{n}}^B \varepsilon\br{\alpha, \frac{\sdpL}{k}} + 2\br{1-\br{1-\frac{kB - 1}{n}}^B} \varepsilon(\alpha, \sdpL),
    \end{aligned}
    \end{equation*}

    In step (b), we note that $2\bP\bs{\cE}\varepsilon\br{\alpha, \frac{\sdpL}{k}} + 2\bP\bs{\neg \cE}\varepsilon\br{\alpha, \sdpL}$ increases with $\bP\bs{\cE}$ by the fact that $\varepsilon\br{\alpha, \frac{\sdpL}{k}}\leq \varepsilon\br{\alpha, \sdpL}$. Then, we obtain the upper bound by substitute in the lower bound of $\bP\bs{\cE}$ from~\Cref{eq:pas3}. 
    This completes the proof. 
\end{proof}

However, under additional structural assumption on the pair of datasets that needs to remain indistinguishable, we can establish a more effective amplification of smooth RDP through subsampling. In~\Cref{app:gradient-based-methods}, we present the additional assumption with~\Cref{assump:dp-sgd-subsampling}. We also apply the stronger amplification in the proof of overall privacy guarantees for DP-SGD with subsampling.


\subsection{Proof of~\Cref{thm:sdp-parameters}(Derivation of~\Cref{tab:sdp-parameters})}\label{app:srdp-table}
In this section, we explore results regarding the SRDP parameters for various DP mechanisms. We start with the formal definitions of two fundamental assumptions, Lipschitzness and smoothness, in \Cref{defn:lipschitzness,defn:smoothness}. Then, we provide the proof of~\Cref{thm:sdp-parameters} (SRDP parameters in~\Cref{tab:sdp-parameters}). For clarity, we present the proof for each private mechanism separately in~\Cref{app:sdp-output-perturbation-sampling,app:gradient-based-methods}.

\begin{defn}[Lipschitzness]\label{defn:lipschitzness}
   A function $f: \cK \rightarrow \reals$ is $L$-Lipschitz over the domain $\cK\subset \cX$ with respect to the distance function $d: \cX\times \cX\rightarrow \reals_+$ if for any $X, X'\in \cK$, $\abs{f(X) - f(X') }\leq Ld(X, X')$. 
\end{defn}
\begin{defn}[Smoothness]\label{defn:smoothness}
    A loss function $\ell: \Theta\times \cK\to \reals_+$ is $\mu$-smooth if for all $\theta\in \Theta$ and $x, x'\in \cK$, $\norm{\nabla_\theta \ell(\theta, x) - \nabla_\theta \ell(\theta, x)}_2\leq \mu\norm{x - x'}_2$. 
    
\end{defn}

We note that 
\Cref{defn:smoothness} is different from the usual definition of smoothness in the literature that requires $\norm{\nabla_\theta \ell(\theta, x) - \nabla_\theta \ell(\theta', x)}_2\leq \mu\norm{\theta - \theta'}_2$ for all $x\in \cK$ and $\theta, \theta'\in \Theta$ \citep{chaudhuri2011differentially,feldman18iteration}. However, many loss functions, including square loss and logistic loss, are smooth on common model classes based on~\Cref{defn:smoothness}. 

\subsubsection{Output Perturbation and Random Sampling methods }\label{app:sdp-output-perturbation-sampling}
This section provides the detailed theorems and proofs for Smooth RDP parameters of output perturbation methods (Gaussian and Laplace mechanism) and random sampling methods (Exponential mechanism) in~\Cref{thm:full-gaussian-sdp,thm:full-laplace-sdp,thm:full-exponential-sdp} respectively, as stated in~\Cref{tab:sdp-parameters}.

\begin{thm}[\Cref{thm:sdp-parameters} for Gaussian mechanism]\label{thm:full-gaussian-sdp}
    For an $L$-Lipschitz function $f$ with global sensitivity $\Delta_f$ and the privacy parameter $\varepsilon > 0$, let $\cM_G(S) = f(S) + \cN\br{0, \nicefrac{\Delta_f^2}{\varepsilon^2}}$ denote the Gaussian mechanism.  For a dataset collection $\cL$ and for $\sdpL > 0$, $\cM_G$ is $\br{\alpha, \frac{\alpha\varepsilon^2}{2}}$-RDP and $\br{\alpha, \frac{\alpha L^2\sdpL^2\varepsilon^2}{2}}$-SRDP for $\alpha \geq 1$. 
\end{thm}
\begin{proof}
    We first show that $\cM_G$ is $\br{\alpha, \frac{\alpha\varepsilon^2}{2}}$-RDP. For any neighboring datasets $S, S'$ with $d_H(S, S') = 1$, 
    \begin{equation*}
    \begin{aligned}
          \alphaRDP{\cM_G(S)}{\cM_G(S')} &= \alphaRDP{f(S) + \cN\br{0, \nicefrac{\Delta_f^2}{\varepsilon^2}}}{f(S') + \cN\br{0, \nicefrac{\Delta_f^2}{\varepsilon^2}}}\\
          &= \alphaRDP{\cN\br{f(S), \nicefrac{\Delta_f^2}{\varepsilon^2}} }{\cN\br{f(S'), \nicefrac{\Delta_f^2}{\varepsilon^2}}}
    \end{aligned}
    \end{equation*}

    We apply~\Cref{lem:gaussian-distributions-renyi-divergence} to upper bound the Rényi divergence between two Gaussian random variable with same variance. 
    \begin{lemL}[Corrolary 3 in \citet{mironov2017renyi}]\label{lem:gaussian-distributions-renyi-divergence}
        For any two Gaussian distributions with the same variance $\sigma^2$ but different means $\mu_0, \mu_1$, denoted by $\cN\br{\mu_0, \sigma^2}$ and $\cN(\mu_1, \sigma^2)$, the following holds, 
        \[D_\alpha\br{\cN(\mu_0, \sigma^2)||\cN(\mu_1, \sigma^2)}\leq \frac{\alpha\br{\mu_0 - \mu_1}^2}{2\sigma^2}.\]
    \end{lemL}

    Hence, 
    \begin{equation*}
             \alphaRDP{\cM_G(S)}{\cM_G(S')} \overset{(a)}{\leq} \sup_{(S, S')\in \cS} \frac{\alpha\br{f(S) - f(S')}^2}{2\nicefrac{\Delta_f^2}{\varepsilon^2}} \overset{(b)}{\leq} \frac{\alpha\varepsilon^2}{2},
    \end{equation*}
    where $(a)$ follows from \Cref{lem:gaussian-distributions-renyi-divergence} and $(b)$ follows due to the fact that $\abs{f(S) - f(S')}\leq \Delta_f$ by the definition of global sensitivity. 
    
    Then, we show that $\cM_G$ is $\br{\alpha, \frac{\alpha L^2\sdpL^2\varepsilon^2}{2}}$-SRDP over the dataset collection $\cL$. Specifically, we will show that for any $\sdpL > 0$, for any two datasets $S, S'\in \cL$ with $d_{12}(S, S')\leq \sdpL$, $\alphaRDP{\cA(S)}{\cA(S')} \leq \frac{\alpha L^2\sdpL^2\varepsilon^2}{2\Delta_f^2}$. Let $S, S'\in \cL$ such that $d_{12}(S, S')\leq \sdpL$, we have 
      \begin{equation}
          \begin{aligned}
              \alphaRDP{\cM_G(S)}{\cM_G(S')} &= \alphaRDP{\cN\br{f(S), \nicefrac{\Delta_f^2}{\varepsilon^2}}}{\cN\br{f(S'), \nicefrac{\Delta_f^2}{\varepsilon^2}}}\\
              &\overset{(a)}{\leq}\frac{\alpha \br{f(S) - f(S')}^2}{2\nicefrac{\Delta_f^2}{\varepsilon^2}}\\
              &\overset{(b)}{\leq}\frac{\alpha\br{L\norm{S - S'}_F}^2\varepsilon^2}{2\Delta_f^2} \overset{(c)}{\leq} \frac{\alpha L^2\sdpL^2\varepsilon^2}{2\Delta_f^2}
          \end{aligned}
      \end{equation}
    where step (a) follows by~\Cref{lem:gaussian-distributions-renyi-divergence}, step (b) follows by the Lipschitzness of the function $f$ with respect to the Frobenius norm, and step (c) follows by the fact that Frobenius norm is always smaller than $L_{12}$ norm, \ie $\norm{S- S'}_F \leq d_{12}(S, S')$. 
\end{proof}
\begin{thm}[\Cref{thm:sdp-parameters} for Laplace mechanism]\label{thm:full-laplace-sdp}
For an $L$-Lipschitz function $f$ with global sensitivity $\Delta_f$ and the privacy parameter $\varepsilon$, let $\cM_L(S) = f(S) + Lap\br{\nicefrac{\Delta_f}{\varepsilon}}$ denote the Laplace mechanism. For a dataset collection $\cL$ and for $\sdpL > 0$, $\cM_L$ is $\br{\alpha, \varepsilon}$-RDP and $\br{\alpha, 
\frac{L\sdpL}{\Delta_f}\varepsilon}$-SRDP for $\alpha \geq 1$. 
\end{thm}

\begin{proof}
Laplace mechanism $\cM_L$ is $(\varepsilon, 0)$-DP~\Citep{dwork06dp}. By the equivalence of $(\varepsilon, 0)$-DP and $(\infty, \varepsilon)$-RDP~\citep{mironov2017renyi}, $\cM_L$ is $(\infty, \varepsilon)$-RDP. This implies that $\cM_L$ is $(\alpha, \varepsilon)$-RDP for any $1 \leq \alpha \leq \infty$ by the monotonicity of Rényi divergence~(\Cref{lem:monotonicity-RDP}). 
    \begin{lemL}[Monotonicity of Rényi divergence \cite{mironov2017renyi}]\label{lem:monotonicity-RDP}
    For $1 \leq \alpha < \beta$, for any two distributions $P, Q$, the following holds \[D_\alpha(P||Q)\leq D_\beta(P||Q). \]
    \end{lemL}
    
    Next, we will show $\cM_L$ is $(\alpha, \nicefrac{L\sdpL\varepsilon}{\Delta_f})$-SRDP over the set $\cL$. Specifically, we first show that $\cM_L$ is $(\infty, \nicefrac{L\sdpL\varepsilon}{\Delta_f}))$-SRDP. Using the monotonicity of Rényi divergence (\Cref{lem:monotonicity-RDP}), we can propagate this property to any $\alpha < \infty$. 

    Let $h$ be any output in the output space of $f$. For any two datasets $S, S'\in \cL$ with $d_{12}(S, S')\leq \sdpL$, 
    \begin{equation}\label{eq:laplace-prob-upper-bound}
        \begin{aligned}
            \frac{\bP(\cM_L(S) = h)}{\bP(\cM_L(S') = h)} &= \frac{\frac{\varepsilon}{\Delta_f}\exp\br{-\frac{\abs{h - f(S)}\varepsilon}{\Delta_f}}}{\frac{\varepsilon}{\Delta_f}\exp\br{-\frac{\abs{h - f(S')}\varepsilon}{\Delta_f}}}\\
           &\leq \exp\br{\frac{\varepsilon}{\Delta_f}\br{\abs{-f(S')  + f(S)}}} \\
           &\overset{(a)}{\leq}\exp\br{\frac{\varepsilon}{\Delta_f}L\norm{S - S'}_F} \overset{(b)}{\leq} \exp\br{\frac{L\sdpL\varepsilon}{\Delta_f}}. 
        \end{aligned}
    \end{equation}
     where $(a)$ follows by the $L$-Lipschitzness of the $f$ with respect to Frobenius norm and $(b)$ follows from the fact that $\norm{S - S'}_F\leq d_{12}(S, S') \leq \sdpL$. 

     This implies the output distributions have bounded infinite Rényi divergence, \ie
    \begin{equation}
        \label{eq:laplace-sdp-bounded-infinite-renyi-divergence}
        D_\infty(\cM_L(S)||\cM_L(S')) = \max_{h\in \cH}\log\frac{\cP(\cM_L(S) = h)}{\cP(\cM_L(S') = h)} \leq \frac{L\sdpL\varepsilon}{\Delta_f},
    \end{equation}
    where the last inequality follows from~\Cref{eq:laplace-prob-upper-bound}. 

    By Monotocity of Rényi divergence~(\Cref{lem:monotonicity-RDP}), for any $\alpha\geq 1$, the $\alpha$-Rényi divergence between the output distribution of Laplace mechanism on any two datasets $S, S'\in \cL$ with $d_{12}(S, S')\leq \sdpL$ is also bounded by $\nicefrac{L\sdpL\varepsilon}{\Delta_f}$. This concludes the proof. 

\end{proof}
\begin{thm}[\Cref{thm:sdp-parameters} for Exponential mechanism]\label{thm:full-exponential-sdp}
    For an $L$-Lipschitz score function $Q$ with global sensitivity $\Delta_Q$, privacy parameter $\varepsilon$, and a dataset collection $\cL$, the exponential mechanism, and for any $\sdpL > 0$, $\cM_E$ is $\br{\alpha, \varepsilon}$-RDP and $\br{\alpha, \nicefrac{L\sdpL\varepsilon}{\Delta_Q}}$-SRDP over $\cL$ for any $\alpha \geq 1$.  
\end{thm}
\begin{proof}
    The proof for RDP of Exponential mechanism is similar to that of Laplace mechanism. As the exponential mechanism $\cM_E$ is $(\varepsilon, 0)$-DP~\citep{dwork06dp}. By the equivalence of $(\varepsilon, 0)$-DP and $(\infty, \varepsilon)$-RDP~\citep{mironov2017renyi}, $\cM_E$ is $(\infty, \varepsilon)$-RDP. By the monotonicity of Rényi divergence (\Cref{lem:monotonicity-RDP}), $\cM_E$ is also $(\alpha, \varepsilon)$-RDP for any $\alpha \geq 1$. 

    To show that $\cM_E$ is $\br{\alpha, \frac{ L\sdpL\varepsilon}{\Delta_Q}}$-SRDP over the set $\cL$, we first show that $\cM_E$ is $(\infty, \nicefrac{L\sdpL\varepsilon}{\Delta_Q})$-SRDP over $\cL$ and then apply the monotonicity of Rényi divergence (\Cref{lem:monotonicity-RDP}). 
    
    For any $S, S'\in \cL$ such that $d_{12}(S, S')\leq \sdpL$ and any output $h$ in the output space,

    \begin{equation}\label{eq:exponential-mechanism-prob-ratio}
        \begin{aligned}
            \frac{\bP(\cM_E(S) = h)}{\bP\br{\cM_E(S') = h}}
            &= \frac{\exp\br{\frac{Q(S, h)\varepsilon}{2\Delta_Q}}}{\exp\br{\frac{Q(S', h)\varepsilon}{2\Delta_Q}}}\frac{\int_{h\in \cH}\exp\frac{Q(S', h)\varepsilon}{2\Delta_Q}dh}{\int_{h\in \cH}\exp\frac{Q(S, h)\varepsilon}{2\Delta_Q}dh}\\
            &\leq \frac{\exp\br{\frac{Q(S, h)\varepsilon}{2\Delta_Q}}}{\exp\br{\frac{Q(S', h)\varepsilon}{2\Delta_Q}}} \frac{\int_{h\in \cH}\exp\frac{Q(S', h) - Q(S, h)\varepsilon}{2\Delta_Q}\exp\frac{Q(S, h)\varepsilon}{2\Delta_Q}dh}{\int_{h\in \cH}\exp\frac{Q(S, h)\varepsilon}{2\Delta_Q}dh}\\
            &\leq\exp\br{\frac{2\max_{h\in \cH}\abs{Q(S, h ) - Q(S', h)}\varepsilon}{2\Delta_Q}}\overset{(a)}{\leq} \exp\br{\frac{L\sdpL\varepsilon}{\Delta_Q}}
        \end{aligned}
    \end{equation}
    where step (a) follows by the $L$-Lipschitzness of the score function $Q(S, h)$ for all $h\in \cH$ and the fact that $\norm{S - S'}_F\leq d_{12}(S, S')\leq \sdpL$. 
    
    This implies bounded infinite order Rényi divergence between the two output distributions, \ie

    \begin{equation}
        \label{eq:exponential-sdp-bounded-infinite-renyi-divergence}
        D_\infty(\cM_E(S)||\cM_E(S')) = \max_{h\in \cH}\log\frac{\cP(\cM_E(S) = h)}{\cP(\cM_E(S') = h)} \leq \frac{L\sdpL\varepsilon}{\Delta_Q},
    \end{equation}
    where the last inequality follows from~\Cref{eq:exponential-mechanism-prob-ratio}. 

    As shown in~\Cref{eq:exponential-sdp-bounded-infinite-renyi-divergence}, $\cM_E$ is $(\infty, \nicefrac{L\sdpL\varepsilon}{\Delta_Q})$-SRDP. This implies that $\cM_E$ is $(\alpha, \nicefrac{L\sdpL\varepsilon}{\Delta_Q})$-SRDP for any $\alpha \geq 1$ by monotonicity of Rényi divergence (\Cref{lem:monotonicity-RDP}). This completes the proof. 
\end{proof}

\subsubsection{Gradient-based methods}\label{app:gradient-based-methods}

In this section, we start with providing the formal definitions of gradient-based methods mentioned in~\Cref{sec:private-mechanisms} and evaluated in~\Cref{tab:sdp-parameters}, including DP-GD, DP-SGD with subsampling, and DP-SGD with iteration. Then, we provide the detailed theorems and proofs for their Smooth RDP parameters, as stated in~\Cref{tab:sdp-parameters}, in~\Cref{thm:full-dpgd-sdp,full-sgdsamp-sdp,thm:dp-sgd-iteration-srdp-parameter}.


\paragraph{DP-GD}Given a $L$-Lipschitz loss function $\ell$, DP-GD \citep{bassily2014private, song21Evading}, denoted by $\cA_{\mathrm{DP-GD}}$ starts from some random initialization $w_0$ in the parameter space $\cH$ and conducts projected gradient descent for $T$ iterations as $w_{t + 1} =\Pi_{\cH}(w_t - \eta \tilde{g}_t)$ with the noisy gradient on the whole dataset $\tilde{g}_t$ defined as $\tilde{g}_t = \frac{1}{n}\sum_{i = 1}^n\nabla_w \ell(w_{t-1}, x) + \cN(0, \sigma^2)$, and outputs the average parameter over the $T$ iterations $\frac{1}{T}\sum_{t = 1}^T w_t$. 

\paragraph{DP-SGD with subsampling}Another commonly used method is DP-SGD with subsampling \citep{Abadi16dpsgd, bassily2014private}, denoted as $\DPSGDsamp
$. In each gradient descent step, DP-SGD with subsampling first draws a uniform subsample of size $B$ from the dataset without replacement. The gradient update is then performed by adding noise to the average gradient derived from this subsample. In contrast, DP-GD computes the average gradient of the entire dataset for its updates.

\paragraph{DP-SGD with iteration} Differentially Private Fixed Gradient Descent (DP-FGD), denoted by $\DPFGD{(T)}$, is a variant of DP-SGD. It processes the data points in a fixed order --- the gradient at the $t^{\text{th}}$ step is calculates using the $t^{\text{th}}$ point in the dataset--- and outputs the parameter obtained after the $T^{\text{th}}$ iteration. \emph{DP-SGD with iteration}~\citep{feldman18iteration}, denoted as $\DPSGDiter$, uses DP-FGD as a base procedure. It takes an extra parameter $\maxD$ and first uniformly samples an integer $T$ from $\bs{n-\maxD + 1}=\bc{1,\ldots,n-\maxD + 1}$ and releases the output from $\DPFGD{(T)}$, \ie, it releases the result from the $T^{\text{th}}$ iteration. While DP-SGD with iteration cannot take advantage of privacy amplification by subsampling for its privacy analysis, it relies on privacy amplification with iteration \citep{feldman18iteration} to achieve a comparable privacy guarantee to that of $\DPSGDsamp$.  


\begin{thm}[\Cref{thm:sdp-parameters} for DP-GD]\label{thm:full-dpgd-sdp}
    For an $L$-Lipschitz and $\mu$-smooth loss function $\ell$, privacy parameter $\varepsilon$, dataset collection $\cL$, and $\sigma = \frac{L\sqrt{T}}{\varepsilon n}$, $\DPGD$ is $\br{\alpha, \frac{\alpha \varepsilon^2}{2}}$-RDP and $\br{\alpha, \frac{\alpha\mu^2 \sdpL^2\varepsilon^2}{2L^2}}$-SRDP over $\cL$ for any $\alpha\geq 1, \sdpL > 0$. 
\end{thm}

\begin{proof}
   We denote each noisy gradient descent step as $\theta_{t+1} = \theta_t - \eta g(\theta_t, S)$ where $g(\theta_t, S) = \frac{1}{n}\sum_{i = 1}^n\nabla_\theta\ell(\theta, S_i) + \cN(0, \nicefrac{L\sqrt{T}}{\varepsilon n})$ is the noisy gradient. To show that each gradient descent step is RDP, it suffices to show the gradient operator $g$ is $(\alpha, \nicefrac{\alpha\varepsilon^2}{2})$-RDP by the fact that RDP is preserved by post-processing (\Cref{lem:post-processing-RDP}). 

   \begin{lemL}[Post-processing of RDP \citep{mironov2017renyi}]
       \label{lem:post-processing-RDP}
       Let $\alpha > 0, \varepsilon: \reals \times \reals \to \reals$, and $f$ be an arbitrary algorithm. For any $\ell > 0$, if $\cA$ is $(\alpha, \varepsilon(\alpha))$-RDP, then $f\circ\cA$ is $(\alpha, \varepsilon(\alpha))$-RDP. 
   \end{lemL}
    
    Now, we show the gradient operator $g$ is $\br{\alpha, \frac{\alpha\varepsilon^2}{2}}$- RDP. By the Lipschitzness of the loss function, for all $\theta, S_i$, we have $\nabla_\theta(\ell(\theta, S_i))\leq L$. Thus, the global sensitivity of the average gradient $\frac{1}{n}\sum_{i = 1}^n\nabla_\theta \ell(\theta, S_i)$ is upper bounded by $\frac{L}{n}$. For any neighboring datasets $S, S'$ such that $d_H(S, S') = 1$, for any $\theta$,

    \begin{equation}
           \alphaRDP{g(\theta, S)}{g(\theta, S')}\overset{(a)}{\leq} \frac{\alpha \br{\frac{1}{n}\sum_{i = 1}^n \nabla_\theta\ell(\theta, S_i) - \frac{1}{n}\sum_{i = 1}^n \nabla_\theta\ell(\theta, S_i')}^2}{2\frac{L^2T}{\varepsilon^2 n^2}} \overset{(b)}{\leq} \frac{2\alpha\varepsilon^2}{T}
    \end{equation}
    where step (a) follows by the Rényi divergence of Gaussian mechanism (\Cref{lem:gaussian-distributions-renyi-divergence}), and step (b) follows by the fact that the sensitivity of each gradient estimation is $\frac{L}{n}$. This implies that each gradient estimation $g$ and thus, each gradient descent step, are $(\alpha, \nicefrac{2\alpha\varepsilon^2}{2
    T})$-RDP. Applying the composition theorem of RDP (\Cref{lem:composition-rdp}), we can show that DP-GD with $T$ gradient descent step is  $\br{\alpha, {2\alpha \varepsilon^2}}$-RDP. 
    \begin{lemL}[Composition of RDP, Proposition 1 in~\citet{mironov2017renyi}]
        \label{lem:composition-rdp}
        For $\alpha \geq 1$, let $f$ be $(\alpha, \varepsilon_1)$-RDP and $g$ be $(\alpha, \varepsilon_2)$-RDP. Then, the mechanism $(f, g)$ is $(\alpha, \varepsilon_1 + \varepsilon_2)$-RDP. 
    \end{lemL}
    
    Then, for any $S, S'\in \cL$ such that $d_{12}(S, S')\leq \sdpL$, we show that the Rényi divergence between the gradient estimate with $S$ and $S'$, $\alphaRDP{g(\theta, S)}{ g(\theta, S')}$, is upper bounded. By the definition of gradient operator $g$ and $\sigma = \frac{L\sqrt{T}}{\varepsilon n}$, we have 
    \begin{equation}
        \begin{aligned}
            \alphaRDP{g(\theta, S)}{g(\theta, S')} &= \alphaRDP{\cN\br{\frac{1}{n}\sum_{i = 1}^n \nabla_\theta \ell(\theta, S_i), \frac{L^2T}{\varepsilon^2n^2}}}{\cN\br{\frac{1}{n}\sum_{i = 1}^n \nabla_\theta \ell(\theta, S_i'), \frac{L^2T}{\varepsilon^2n^2}}}\\
            &\overset{(a)}{\leq} \frac{\alpha \norm{\frac{1}{n}\sum_{i = 1}^n\nabla_\theta \ell(\theta, S_i) - \frac{1}{n}\sum_{i = 1}^n\nabla_\theta\ell(\theta, S_i')}_2^2}{2\frac{L^2 T}{n^2\varepsilon^2}}\\
            &\overset{(b)}{\leq} \frac{\alpha \varepsilon^2 \mu^2 \br{\sum_{i = 1}^n \norm{S_i - S'_i}_2}^2}{2L^2 T} \overset{(c)}{=} \frac{\alpha\varepsilon^2\mu^2\sdpL^2}{2L^2T}
        \end{aligned}
    \end{equation}
    where step (a) follows from~\Cref{lem:gaussian-distributions-renyi-divergence}, step (b) follows from the smoothness assumption of the loss function, \ie $\norm{\nabla_\theta \ell(\theta, S_i) - \nabla_\theta\ell(\theta, S_i')}_2\leq \mu\norm{S_i - S_i'}_2$, and step (c) follows from the definition of $L_{12}$ distance and the fact that $d_{12}(S, S')\leq \sdpL$. 
    
    Applying composition of SRDP (\Cref{lem:properties-srdp}) over the $T$ gradient descent steps concludes the proof. 
\end{proof}


For establishing the overall privacy guarantees for DP-SGD with subsampling and iteration, we introduce two properties of the dataset collection: the inverse point-wise distance (\Cref{assump:dp-sgd-subsampling}) and the maximum distance (\Cref{assump:dp-sgd-iteration}). Then, we present the overall privacy guarantee for DP-SGD with subsampling and iteration in~\Cref{thm:full-dpgd-sdp} and~\Cref{thm:dp-sgd-iteration-srdp-parameter}, with additional assumptions on the inverse point-wise distance and $\sdpL$-constrained maximum distance of the dataset collection.

      \begin{defn}[Inverse point-wise distance of a dataset collection $\cL$]\label{assump:dp-sgd-subsampling}
        Let $k$ be the maximum integer such that for every pair of datasets $S_1 = \bc{S_1^i}_{i = 1}^n\in \cL$ and $S_2 = \bc{S_2^i}_{i = 1}^n\in \cL$ and for all $i\in [n]$, $\norm{S_1^i - S_2^i}_2\leq \frac{d_{12}\br{S_1, S_2}}{k}$. The inverse point-wise distance $\gamma$ of a dataset collection $\cL$ is defined as $\gamma = \frac{n}{k}$.
    \end{defn}
  
     \begin{defn}[$\sdpL$-constrained maximum distance of a dataset collection $\cL$]
        \label{assump:dp-sgd-iteration}
        For $\sdpL>0$, the $\tau$-constrained maximum distance $\maxD$ of a dataset collection $\cL$ is defined as the maximum Hamming distance between any two datasets $S_1, S_2\in \cL$ such that $d_{12}(S_1, S_2) \leq \sdpL$. 
    \end{defn}
    \begin{thm}[\Cref{thm:sdp-parameters} for DP-SGD with subsampling]\label{full-sgdsamp-sdp}
        For any $L$-Lipschitz and $\mu$-smooth loss function $\ell$, privacy parameter $\varepsilon$, dataset collection $\cL$ with inverse point-wise distance $\gamma$, and $\sigma = \Omega\br{\frac{L\sqrt{T}}{\varepsilon n}}$, $\DPSGDsamp$ is $\br{\alpha, \frac{\alpha^2\varepsilon^2}{2}}$-RDP and $\br{\alpha, \frac{\alpha\mu^2\sdpL^2\varepsilon^2\gamma^2}{2L^2}}$-SRDP over $\cL$ for any $\sdpL > 0$ and $1\leq \alpha\leq \min\bc{\frac{\sqrt{T}}{\varepsilon}, \frac{L^2T}{\varepsilon^2 n^2}\log \frac{n^2\varepsilon}{L\sqrt{T}}}$.
    \end{thm}
    \begin{proof}
        We let $\pi_u$ be a subsampling mechanism that uniformly samples one point from the dataset. We denote each gradient descent step as $\theta_{t+1} = \theta_t - \eta g(\theta_t, S)$ where $g(\theta_t, S) = \nabla_\theta\ell(\theta, \pi_u(S)) + \cN(0, \sigma^2)$ represents the gradient estimate. 
        
        We will show that each gradient estimate $g$ is $(\alpha, \nicefrac{\alpha^2\varepsilon^2}{2T})$-RDP and thus, each gradient descent step satisfies RDP with the same parameters by post-processing theorem of RDP (\Cref{lem:post-processing-RDP}). Then, by the composition theorem for RDP (\Cref{lem:composition-rdp}) over the $T$ gradient descent steps, we can conclude that $\DPSGDsamp$ is $(\alpha, \nicefrac{\alpha^2\varepsilon^2}{2})$-RDP.   

        First, for each gradient estimate $g(\theta, \pi_u(S))$, we apply~\Cref{lem:dp-sgd-lemma} to upper bound the Rényi divergence for each gradient step in~\Cref{eq:appl-lem-3-abadi}. 

        \begin{lemL}[Lemma 3 in~\citet{Abadi16dpsgd}]\label{lem:dp-sgd-lemma}
        Suppose that $f: \cX\rightarrow \reals^p$ is a function with $\norm{f(\cdot)}_2 \leq L$.  Assume $\sigma \geq 1$, and let $i \sim \mathrm{Unif}([n])$ represent a uniform random variable over the integers in the set $[n] = {1, 2, \ldots, n}$. Then, for any positive integer $1\leq \alpha \leq \sigma^2 \ln \frac{n}{\sigma}$ and any pair of neighboring datasets $S, S'$, the mechanism $\cM(S) = f(S_i) + \cN(0, \sigma^2I)$ satisfies 
            \begin{equation}\label{eq:dp-sgd-lemma}
                \alphaRDP{\cM(S)}{\cM(S')} \leq \underbrace{\frac{L^2\alpha(\alpha +1)}{n^2(1-\nicefrac{1}{n})\sigma^2}}_{\mathrm{(I)}} + \underbrace{\bigO{\frac{\alpha^3L^3}{n^3\sigma^3}}}_{\mathrm{(II)}}.
            \end{equation}
        \end{lemL}
        
        \noindent For neighboring datasets $S, S'$ and for any $\theta$, each gradient estimate satisfies, 
        \begin{equation}\label{eq:appl-lem-3-abadi}
            \alphaRDP{g(\theta, \pi_u(S)) }{g(\theta, \pi_u(S'))}  \overset{(a)}{\leq}  \frac{c \alpha^2 L^2}{n^2 \sigma^2} \overset{(b)}{\leq} \frac{\alpha^2\varepsilon^2}{2T},
        \end{equation}
        Note that part (II) is smaller than part (I) in~\Cref{eq:dp-sgd-lemma} for $\alpha \leq \frac{\sqrt{T}}{\varepsilon}$. Hence, step (a) follows by application of~\Cref{lem:dp-sgd-lemma} for some positive constant $c \geq 1$. Step (b) follows by choosing $\sigma = \frac{L\sqrt{cT}}{\varepsilon n}$. 
        
        Next, we will show that $\DPSGDsamp$ is $\br{\alpha, \frac{\alpha \sdpL^2\varepsilon^2\mu^2\gamma^2}{2L^2}}$-SRDP. It suffices to show that each gradient estimate is $\br{\alpha, \frac{\alpha \sdpL^2\varepsilon^2\mu^2\gamma^2}{2TL^2}}$-SRDP by the post-processing and composition theorem of SRDP (\Cref{lem:properties-srdp}). 

        Let $k$ be the maximum integer such that for every pair $S, S'\in \cL$ with $d_{12}(S, S')\leq \sdpL$, the point-wise distance is less than $\sdpL/k$. We note that for any $S, S'\in \cL$ with $d_{12}(S, S')\leq \sdpL$, then for any $i\in [n]$, 
        \begin{equation}
            \label{eq:sgd-subsampling-srdp}
            d_{12}(\pi_u(S), \pi_u(S')) = \norm{S_i-S_i'}_2\leq \frac{\sdpL}{k}.
        \end{equation}

        Next, we will upper bound the Rényi divergence between two gradient estimates. Let \(i\sim \mathrm{Unif}\br{\bs{n}}\) be the sampled index. Then, for any $S, S'\in \cL$ such that $d_{12}(S, S')\leq \sdpL$, for any $\theta$,
        \begin{equation}\label{eq:dp-sgd-subsampling-with-assumption}
        \begin{aligned}
            \alphaRDP{g(\theta, \pi_u(S))}{g(\theta, \pi_u(S'))} &\overset{(a)}{\leq} \alphaRDP{g(\theta, S_i)}{g(\theta, S_i')}\\
            &\overset{(b)}{\leq} \frac{\norm{\nabla_\theta\ell(\theta, S_i) - \nabla_\theta\ell(\theta, S_i')}_2^2}{2\sigma^2}\\
            &\overset{(c)}{\leq} \frac{\sdpL^2\alpha \mu^2}{2k^2\sigma^2 }
            \overset{(d)}{\leq}\frac{\alpha\sdpL^2\varepsilon^2\mu^2}{2TL^2} \br{\frac{n}{k}}^2 \overset{(e)}{=}\frac{\alpha\sdpL^2\varepsilon^2\mu^2\gamma^2}{2TL^2} 
        \end{aligned}
    \end{equation}

    In step (a), $i\sim \text{Unif}([n])$ is as defined above. Step (b) follows by the Rényi divergence of Gaussian distributions (\Cref{lem:gaussian-distributions-renyi-divergence}), and step (c) follows by the smoothness assumption of $\ell$, \ie $\norm{\nabla_\theta\ell(\theta, S) - \nabla_\theta\ell(\theta, S')}_2\leq \mu\norm{S - S'}_2$ and~\Cref{eq:sgd-subsampling-srdp}. Step (d) follows by the substitution of $\sigma = \frac{L\sqrt{cT}}{\varepsilon n}$ and $c\geq 1$, and step (e) follows by the definition of inverse point-wise distance $\gamma = \frac{n}{k}$ as specified in~\Cref{assump:dp-sgd-subsampling}. 

    This completes the proof.  
    \end{proof}

    \begin{remark}
        Without~\Cref{assump:dp-sgd-subsampling}, we can employ the weaker subsampling results for SRDP (\Cref{lem:full-pas-srdp}). However, that would results in an extra $n^{3/2}$ factor in the SRDP parameter, as~\Cref{eq:dp-sgd-subsampling-with-assumption} will be replaced with~\Cref{eq:dp-sgd-subsampling-without-assumption}. 
 
    \begin{equation}\label{eq:dp-sgd-subsampling-without-assumption}
        \begin{aligned}
            \alphaRDP{g(\theta, \pi_u(S))}{g(\theta, \pi_u(S'))}&\leq\min_{k\geq 1} \bs{\frac{1}{k^2}\br{1-\frac{k-1}{n}} + \frac{k-1}{n}}\frac{\alpha\mu^2\sdpL^2\varepsilon^2n^2}{L^2T}\\
             &\overset{(a)}{\leq} \br{\frac{1}{2n-3} + \frac{(2n-4)\br{\sqrt{2n-3}-1}}{n(2n-3)}} \frac{\alpha\mu^2\sdpL^2\varepsilon^2n^2}{L^2T}\\
             &\overset{(b)}{\leq}
             \frac{2\mu^2\sdpL^2\varepsilon^2n^{3/2}\alpha}{L^2T}
        \end{aligned}
    \end{equation}
    where $(a)$ is obtained by choosing $k = \sqrt{2n-3}$ and $(b)$ follows by $\br{\frac{1}{2n-3} + \frac{(2n-4)\br{\sqrt{2n-3}-1}}{n(2n-3)}} \leq \frac{2}{\sqrt{n}}$. 
   \end{remark}
   
    \begin{thm}[\Cref{thm:sdp-parameters} for DP-SGD with iteration]\label{thm:dp-sgd-iteration-srdp-parameter}
        For an $L$-Lipschitz and $\mu$-smooth convex loss function $\ell$, let $\beta = \sup_{z\in \cX\times \cY}\norm{\nabla_\theta^2\ell(\theta; z)}$. For any privacy parameter $\varepsilon$, learning rate $\eta\leq 2/\beta$, $\sdpL\geq 0$, dataset collection $\cL$ with $\sdpL$-constrained maximum distance $\maxD$, 
         and $\alpha>1$ such that $\max\bc{L\sqrt{2(\alpha-1)\alpha}, \sdpL L\sqrt{2{(\alpha - 1)\alpha}}} < \sqrt{\frac{2\log n}{n}}\frac{8 L}{\varepsilon}$, assume $\cL$ satisfies that for any $S, S'\in \cL$ with $d_{12}(S, S')\leq \sdpL$, the differing points in $S, S'$ are consecutive, then $\DPSGDiter$ with parameter $\maxD$ and $\sigma = \sqrt{\frac{2\log n}{n}}\frac{8 L}{\varepsilon}$ is $\br{\alpha, \frac{\alpha \varepsilon^2}{2}}$-RDP and $\br{\alpha, \frac{\alpha \sdpL^2\mu^2 n\log \br{n - \maxD +2}}{2(n-\maxD+1)L^2\log n}}$-SRDP. 
    \end{thm}
    \begin{proof}
        We first note that $\DPSGDiter$ is $(\alpha, \frac{\alpha\varepsilon^2}{2})$-RDP following Theorem 26 in~\citet{feldman18iteration} (see \Cref{lem:dp-sgd-iteration-privacy} below for completeness).

        \begin{lemL}[Privacy guarantee of SGD by iteration, Theorem 26 in~\citet{feldman18iteration}]\label{lem:dp-sgd-iteration-privacy}
            Let $\ell$ be an convex $L$-Lipschitz and $\beta$-smooth loss function over $\reals^d$. Then, for any learning rate $\eta\leq \frac{2}{\beta}$, $\alpha > 1$, $\sigma \geq L\sqrt{2(\alpha -1 )\alpha}$, $\DPSGDiter$ satisfies $\br{\alpha, \frac{4\alpha L^2\log n}{n\sigma^2}}$-RDP.  
        \end{lemL}
        
        For SRDP, we consider datasets $S, S'\in \cL$ such that $d_{12}(S, S')\leq \sdpL$ and with all differing points appearing consecutively. Without loss of generality, we assume the first differing point has index $t$. By the assumption on the $\sdpL$-constrained maximum distance of $\cL$, there are in total $\maxD$ consecutive differing points in $S, S'$, $t \leq n-\maxD + 1$. In the following, we consider two cases: i) $t > T$, and ii) $t\leq T$. 

        As the gradient descent step with $S$ and $S'$ are exactly the same before $t$, in the first case, with $t> T$, we have 
        \[\alphaRDP{\DPSGDiter(S)}{\DPSGDiter(S')} = \alphaRDP{\cA_{FGD}^T(S)}{\cA_{FGD}^T}\leq \alphaRDP{\cA_{FGD}^t(S)}{\cA_{FGD}^t} = 0.\] 

        In the second case, we first employ~\Cref{lem:dp-fgd-guarantee} to upper bound the Rényi divergence of the output of DPFGD at some fixed step $T$. 
        
        \begin{lem}\label{lem:dp-fgd-guarantee}
            For $\tau \geq 0$, let $\cL$ be a dataset collection with dataset size $n$ and $\tau$-constrained maximum distance $\maxD$. Let $\cL$, $\sdpL$ and $\sigma$ be parameters that satisfy the same assumptions as in~\Cref{thm:dp-sgd-iteration-srdp-parameter}. For any two datasets $S, S'\in \cL$ with $d_{12}(S, S')\leq \sdpL$, the algorithm $\DPFGD{T}$ satisfies  
            \[\alphaRDP{\DPFGD{T}(S)}{\DPFGD{T}(S')}\leq \frac{2\alpha\sdpL^2 L}{\sigma^2(n-\tilde{t}+1)},\]
            where $\tilde{t}\leq n - \maxD + 1$ being the index of the first pair of differing points. 
        \end{lem}   
        For some fixed $T$, by~\Cref{lem:dp-fgd-guarantee} the Rényi divergence between the outputs of $\DPSGDiter$ on datasets $S$ and $S'$ is upper bounded by 
        \begin{equation}
            \label{eq:dp-fgd-guarantee1}
            \alphaRDP{\cA_{FGD}^T(S)}{\cA_{FGD}^T(S')}\leq\frac{2\alpha\sdpL^2L}{\sigma^2(T-t+1)}.
        \end{equation}

        Then, as $T$ is a uniform random variable in \(\DPSGDiter(S)\), we can upper bound the R\'enyi divergence at some random time $T$ by the weak convexity of R\'enyi divergence (\Cref{lem:weak-convexity-renyi-divergence}). We note that $t\leq T$ in case (ii), then for all $T$, as $\alpha$ satisfies $\sigma \geq \sdpL L\sqrt{2{(\alpha - 1)\alpha}
        }$, 
        \[\alphaRDP{\cA_{FGD}^T(S)}{\cA_{FGD}^T(S')} \leq \frac{2\alpha\sdpL^2L^2}{\sigma^2\br{T - t + 1}} \leq \frac{2\alpha \sdpL^2L^2}{\sigma^2}\leq \frac{1}{\alpha-1}.\] 
    
        Therefore, we can apply~\Cref{lem:weak-convexity-renyi-divergence} with $c = 1$. For any $t\leq T$,
        \begin{equation}
            \begin{aligned}
                \alphaRDP{\DPSGDiter(S)}{\DPSGDiter(S')}
                &\overset{(a)}{\leq} \frac{2}{n-\maxD+1}\sum_{T\in [n-\maxD+1]}\alphaRDP{\DPFGD{T}(S)}{\DPFGD{T}(S')}\\
                &\overset{(b)}{\leq}\frac{2}{n-\maxD + 1}\sum_{T = t}^{n-\maxD +1}\frac{2\alpha\sdpL^2\mu^2}{\sigma^2\br{T - t + 1}}\\
                &\leq \frac{4\alpha\sdpL^2\mu^2}{(n-\maxD +1)\sigma^2}\log \br{n-\maxD-t+2}\\
                &\leq \frac{4\alpha\sdpL^2\mu^2}{\sigma^2}\frac{\log (n-\maxD+2)}{n-k+1}
            \end{aligned}
        \end{equation}
        where $(a)$ follows applying~\Cref{lem:weak-convexity-renyi-divergence} and $(b)$ follows from~\Cref{eq:dp-fgd-guarantee1} for $T\geq t$. Substituting $\sigma$ concludes the proof. 

    \end{proof}

    \begin{proof}[Proof of~\Cref{lem:dp-fgd-guarantee}]
        Consider two datasets $S, S' \in \cL$ that satisfy the assumptions in~\Cref{lem:dp-fgd-guarantee}. Define the point-wise distance between these datasets as $d_i = \norm{S_i - S_i'}_2$. Let $\tilde{t}$ be the first index where $d_i > 0$. By the assumption that differing points in \(S,S^{\prime}\) are consecutively ordered, it follows that $\tilde{t}\leq n- \maxD + 1$. By the definition of $\sdpL$-constrained maximum distance (\Cref{assump:dp-sgd-iteration}) and the assumptions on the dataset collection $\cL$, we have, 
        \begin{equation}\label{eq:di-definition}
            d_i = \frac{\sdpL}{\maxD}, \quad \forall i \in \{\tilde{t}, \tilde{t}+1, \ldots, \tilde{t}+\maxD\}
        \end{equation}
        and $d_i = 0$ otherwise. 

        
        We will define Contractive Noisy Iteration (CNI) (\Cref{defn:cni}) and construct a CNI that outputs $\DPFGD{T}$ after $T$ steps.

        \begin{defn}[Contractive Noisy Iteration (CNI)]
            \label{defn:cni}
            Given an initial state $\theta_0\in \Theta$, a sequence of contractive functions $\psi_t:\Theta\to\Theta$, and a noise parameter $\sigma > 0$, the Contractive Noisy Iteration $(\theta_0, \bc{\theta_t}, \cN(0, \sigma^2))$ is defined by the following update rule: \[\theta_{t+1} = \psi_{t+1}(\theta_t) + Z_t,\]
            where $Z_t\sim \cN(0, \sigma^2)$. 
        \end{defn}

        For $S, S'\in \cL$, we construct two series of contractive function $\bc{\psi_i}$ and $\bc{\psi_i'}$ as the gradient descent on the $i^{\text{th}}$ data point of $S$ and $S'$ respectively. Formally,
        \begin{equation}\label{defn:psi-psi'}
        \begin{aligned}
            \psi_i(\theta) &= \theta - \eta \nabla_\theta(\theta, S_i)\\
            \psi_i'(\theta) &= \theta - \eta \nabla_\theta(\theta, S_i')
        \end{aligned}
        \end{equation}

        The functions $\psi_i$ and $\psi_i'$ are contractive functions for $\eta \leq 2/\beta$~\citep{Nesterov04}. It follows by the definition of DPFGD that $\theta_T = \DPFGD{T}(S)$ and $\theta_T' = \DPFGD{T}(S')$ are the $T^{\text{th}}$ outputs of the CNIs $(\theta_0, \bc{\psi_t}, \cN(0, \br{\eta\sigma}^2))$ and $(\theta_0, \bc{\psi_t'}, \cN(0, \br{\eta\sigma}^2))$ respectively.

        For these two CNIs, we can apply \Cref{lem:shift-reduction} to upper bound the R\'enyi divergence between their $T^{\text{th}}$ outputs.
        \begin{lemL}[Theorem 22 in~\citet{feldman18iteration} with fixed noise distribution]
            \label{lem:shift-reduction}
            Let $X_T$, $X_T'$ denote the output of two Contractive Noisy Iteration $(X_0, \{\psi_t\}, \cN(0, (\eta\sigma)^2))$ and $(X_0, \{\psi_t'\}, \cN(0, (\eta\sigma)^2))$ after \(T\) steps. Let $s_t = \sup_x \norm{\psi_t(x) - \psi_t'(x)}_2$. Let $a_1, \ldots, a_T$ be a sequence of reals such that $z_t = \sum_{i \leq t}s_i - \sum_{i\leq t}a_i\geq 0$ for all $t < T$ and $z_T = 0$. Then, 
            \[\alphaRDP{X_T}{X_T'}\leq\sum_{i = 1}^T\alphaRDP{\cN(0, \eta^2\sigma^2)}{\cN(a_i, \eta^2\sigma^2)}\]
            
        \end{lemL}

        
        Following the definition in~\Cref{lem:shift-reduction}, for the two contractive noisy maps $(\theta_0, \bc{\psi_t}, \cN(0, \br{\eta\sigma}^2))$ and $(\theta_0, \bc{\psi_t'}, \cN(0, \br{\eta\sigma}^2))$, we define $s_i$ as   
        \begin{equation}
            \label{eq:contractive-map-gradient-descent}
            \begin{aligned}
                s_i &= \sup_{\theta}\norm{\psi_i(\theta) - \psi_i'(\theta)}_2 \\
                &\overset{(a)}{\leq} \sup_{\theta}\norm{\theta - \eta g(\theta, S_i) - \theta + \eta g(\theta, S_i')}_2\\
                &\overset{(b)}{\leq} \eta \mu d_i \overset{(c)}{=} \begin{cases} 
                      \frac{\eta \mu\sdpL}{\maxD} & i\in \{\tilde{t}, \tilde{t}+1, \ldots, \tilde{t}+\maxD\}\\
                      0 & \text{Otherwise}
                   \end{cases}
            \end{aligned}
        \end{equation}
        where step (a) follows by~\Cref{defn:psi-psi'}, step (b) follows by the smoothness assumption on the loss function $\ell$, \ie $\norm{\nabla_\theta\ell(\theta, S_i) - \nabla_\theta\ell(\theta, S_i')}_2\leq \mu \norm{S_i - S_i'}_2$ for any $\theta, S_i, S_i'$. Step (c) follows by~\Cref{eq:di-definition}.

        By choosing $a_i = 0$ for all $i < \tilde{t}$ and $a_i = \frac{\eta\sdpL L}{T - \tilde{t} + 1}$ for $i \geq \tilde{t}$, we have that for all $t \leq T-\maxD+1$, $z_t = \sum_{i\leq t}s_i - \sum_{i \leq t} a_i \geq 0$ and $z_T = 0$. This allows for the application of \Cref{lem:shift-reduction}, 
        \begin{equation}\label{eq:dp-sgd-iteration}
            \begin{aligned}
            \alphaRDP{\DPFGD{T}(S)}{\DPFGD{T}(S)} & \overset{(a)}{\leq} \frac{2\alpha}{\eta^2 \sigma^2}\sum_{i = 1}^T a_i^2\\
            &\overset{(b)}{\leq} \frac{2\alpha}{\eta^2\sigma^2}\sum_{i = \tilde{t}}^T\frac{\sdpL^2\eta^2\mu^2}{(T-\tilde{t}+1)^2}\\
            &= \frac{2\alpha\sdpL^2\mu^2}{\sigma^2\br{T - \tilde{t} + 1}}
            \end{aligned}
        \end{equation}
    where step (a) follows from the application~\Cref{lem:shift-reduction} and R\'enyi divergence of two Gaussian distributions (\Cref{lem:gaussian-distributions-renyi-divergence}). Step (b) follows by the definition of $a_i$. 
    \end{proof}

\clearpage

\section{Proof of Meta-Theorem (\Cref{thm:general-thm-training-only}) }
\label{app:sec3-proofs}

\generalTrainingOnly*
\begin{proof}
Consider two neighboring datasets $S_1$ and $S_2$, where $S_1 = S\cup \{z_1\}$ and $S_2 = S \cup\{z_2\}$. Let $\pi_1$ and $\pi_2$ be the pre-processing functions output by the pre-processing algorithm $\pi$ on $S_1$ and $S_2$ respectively. Our objective is to establish an upper bound on the Rényi divergence between the output distribution of $\cA$ on the pre-processed dataset $\pi_1(S_1)$ and $\pi_2(S_2)$, \ie~$\alphaRDP{\cA(\pi_1(S_1))}{\cA(\pi_2(S_2))}$ and $\alphaRDP{\cA(\pi_2(S_2))}{\cA(\pi_1(S_1))}$. 

In the following, we first derive an upper bound on the R\'enyi divergence between $\cA(\pi_1(S_1))$ and $\cA(\pi_2(S_2))$. To do so, we construct a new dataset $\tilde{S}$ using the components of $\pi_1(S_1)$ and $\pi_2(S_2)$ as indicated in~\Cref{fig:thm1-proof-sketch}, \ie $\tilde{S} = \pi_1(S_1)\cup \pi_2(S_2)$. Then, using the same approach, we will upper bound the R\'enyi divergence between $\cA(\pi_2(S_2))$ and $\cA(\pi_1(S_1))$.

\begin{figure}[t]
    \centering
    \includegraphics[width = 0.5\linewidth]{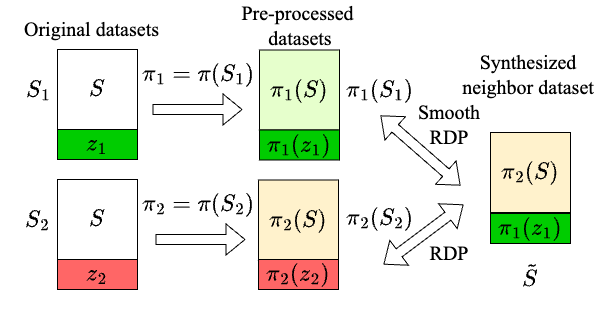}
    \caption{Illustration of the privacy analysis}
    \label{fig:thm1-proof-sketch}
\end{figure}



By contruction, $\tilde{S}$ and $\pi_2(S_2)$ are neighboring datasets, and that the $L_{12}$ distance between $\tilde{S}$ and $\pi_2(S_2)$ is upper bounded by $\Delta_2\Delta_\infty$. Using the RDP property of algorithm $\cA$ we upper bound the divergence between $\cA(\tilde{S})$ and $\cA(\pi_2(S_2))$, 
\begin{equation}
    \label{eq:general-results-rdp-srdp-bounds-1}
        \alphaRDP{\cA(\pi_1(S_1))}{\cA(\tilde{S})}\leq \tilde{\varepsilon}(\alpha, \Delta_\infty\Delta_2).
\end{equation}
Similarly, using the SRDP property of the algorithm $\cA$ over $\cL$, we upper bound the divergence between $\cA(\tilde{S})$ and $\cA(\pi_1(S_1))$,

\begin{equation}
    \label{eq:general-results-rdp-srdp-bounds-2}
        \alphaRDP{\cA(\tilde{S})}{\cA(\pi_2(S_2))} \leq \varepsilon (\alpha).
\end{equation}

\noindent Now, we combine~\Cref{eq:general-results-rdp-srdp-bounds-1,eq:general-results-rdp-srdp-bounds-2} using the weak triangle inequality of Rényi divergence (\Cref{lem:triangle-inequality-renyi-divergence}), to upper bound the R\'enyi divergence between \(\cA(\pi_1(S_1))\) and \(\cA(\pi_2(S_2)\).
\begin{lemL}[Triangle inequality of Rényi divergence \cite{mironov2017renyi}]
\label{lem:triangle-inequality-renyi-divergence}
Let $\mu_1, \mu_2, \mu_3$ be distributions with the same support. Then, for $\alpha > 1$, $p, q > 1$ such that $\frac{1}{p} + \frac{1}{q} = 1$, it holds that \[D_\alpha(\mu_1||\mu_2) \leq \frac{\alpha - \nicefrac{1}{p}}{\alpha - 1}D_{p\alpha}(\mu_1||\mu_3) + D_{q(\alpha - \nicefrac{1}{p})}(\mu_3||\mu_2).\]
\end{lemL}

\noindent Using~\Cref{lem:triangle-inequality-renyi-divergence}, for any $c_1 > 1$, we have
\begin{equation}
    \label{eq:general-results-part1}
    \begin{aligned}
        D_\alpha(\cA(\pi_1(S_1))||\cA(\pi_2(S_2))) &\leq \frac{\alpha - \frac{1}{c_1}}{\alpha - 1}D_{c_1\alpha}(\cA(\pi_1(S_1))|\cA(\tilde{S})) + D_{\frac{c_1\alpha - 1}{c_1-1}}(\cA(\tilde{S})||\cA(\pi_2(S_2)))\\
        &\overset{(a)}{\leq}\frac{\alpha c_1 - 1}{c_1(\alpha - 1)}\tilde{\varepsilon}(\alpha c_1, \Delta_\infty\Delta_2) + \varepsilon\br{\frac{c_1\alpha - 1}{c_1 - 1}}
    \end{aligned}
\end{equation}
where step (a) follows from~\Cref{eq:general-results-rdp-srdp-bounds-1,eq:general-results-rdp-srdp-bounds-2}. 

Similarly, by constructing a dataset $\tilde{S}$ consisting of $\pi_1(S)$ and $\pi_2(z_2)$, we upper bound $\alphaRDP{\cA(\pi_2(S_2))}{\cA(\tilde{S})}$ and $\alphaRDP{\cA(\tilde{S})}{\cA(\pi_1(S_1))}$ with the SRDP and RDP property in a similar manner as~\Cref{eq:general-results-rdp-srdp-bounds-1,eq:general-results-rdp-srdp-bounds-2}. Applying~\Cref{lem:triangle-inequality-renyi-divergence}, we can show that for any $c_2 > 1$, 

\begin{equation}
    \label{eq:general-results-part2}
        D_\alpha(\cA(\pi_2(S_2))||\cA(\pi_1(S_1)))\leq \frac{\alpha c_2- 1}{c_2(\alpha - 1)}\varepsilon(\alpha c_2) + \tilde{\varepsilon}\br{\frac{c_2\alpha - 1}{c_2 - 1}, \Delta_\infty\Delta_2}
\end{equation}

Combining~\Cref{eq:general-results-part1} and~\Cref{eq:general-results-part2} concludes the proof. 
\end{proof}


\clearpage

\section{Proofs for sensitivity of pre-processing algorithms (\Cref{sec:sensitivity-preproc})}
\label{app:sec-diff-processing}

In this section, we bound the sensitivity of pre-processing algorithms discussed in~\Cref{sec:sensitivity-preproc}.

\subsection{Sensitivity analysis of deduplication and quantization}
\deduplicationSensitivities*
\begin{proof}
    Consider two neighboring datasets $S_1, S_2\in \cL$. Without loss of generality, let $S = S_1\cap S_2$, $z_1 = S_1\setminus S$ and $z_2 = S_2\setminus S$, similar to the notations of original datasets in~\Cref{fig:thm1-proof-sketch}.

    \paragraph{Sensitivity analysis of deduplication} As the data space $\cX$ is bounded by $1$, it is obvious that the upper bound on $L_2$ sensitivity of $\pi_\eta^d$ is $1$. 
    
    To bound the \(L_{\infty}\) sensitivity, we first recall the definition of a ``good'' cluster. For a dataset \(S\) and a point \(x \in S\), define \(B(x, \eta; S) = \{\tilde{x} \in S: \norm{\tilde{x}-x}_2 \leq \eta\}\), a ball of radius \(\eta\) around \(x\). A point $x$ is the centroid of a \textit{good cluster} if \(B(x, \eta; S) = B(x, 3\eta; S)\). The set of all good clusters in a dataset $S$ is denoted by \(B(S) = \bc{B_i}_{i = 1}^{m}\)  where \(B_i := B(x_i, \eta; S)\) for all $x_i \in S$ satisfying \(B(x_i, \eta; S) = B(x_i, 3\eta; S)\).

    We will first prove that the difference between the datasets $\pi_{\eta, S_1}^d(S)$ and $\pi_{\eta, S}^d(S)$ is the ball $B(z_1, \eta;  S_1)$. Similarly, we show that the maximum difference between $\pi_{\eta, S_2}^d(S)$ and $\pi_{\eta, S}^d(S)$ is $B(z_2, \eta; S_2)$. Taking supremum over all neighboring datasets in $\cL$, these two results imply that the $L_\infty$ sensitivity of deduplication is upper bounded by twice the size of the largest good cluster in any dataset $S\in \cL$. 
    
    To calculate that the maximum difference between the datasets $\pi_{\eta, S_1}^d(S)$ and $\pi_{\eta, S}^d(S)$, we assume without loss of generality that there exists $x_1, \ldots, x_k\in S$ satisfying $\norm{z_1-x_i}\leq \eta$ for all $i\in [k]$. We consider the following cases: 

    \paragraph{Case I: $x_1, \ldots x_k$ are not in any good cluster centered at some $c\in S$.} We will discuss the two sub-cases: one where the point $z_1$ forms the centroid of a good cluster, and another where it does not. 
    
    If the point $z_1$ is the centroid of a good cluster, then $x_1, \ldots x_k$ are the only points in the good cluster $B(z_1, \eta; S)$ and will be removed by $\pi_{\eta, S_1}^d$. In contrast, in Case I, $x_1, \ldots, x_k$ will not be removed by $\pi_{\eta, S}^d$. Hence, the difference between $\pi_{\eta, S_1}^d(S)$ and $\pi_{\eta, S}^d(S)$ is $\bc{x_1, \ldots, x_k}\subset B(z_1, \eta; S)$. 

    If the point $z_1$ is not the centroid of a good cluster, then $z_1$ is not in any good cluster. We will prove this claim by contradiction. Assume $z_1$ is in a good cluster centered at some $c\in S$. Then, for any $x_i$, $i\in [k]$, 
    \begin{equation}
        \norm{x_i - c}_2\leq \norm{x_i - z_1}_2 + \norm{z_1 - c}_2\leq 2\eta. 
    \end{equation}

    If $\norm{x_i - c}_2 \leq \eta$, then $x_i$ is also in the good cluster around $c\in S$, contradicting the assumption that none of $\bc{x_1, \ldots x_k}$ is in any good cluster. On the other hand, if $\eta \leq \norm{x_i - c}_2 \leq 2\eta$, then $B(c, 3\eta; S)$ cannot be a good cluster. 

    Therefore, we have shown that when $x_1, \ldots x_k$ are not in any good cluster and the point $z_1$ is not the centroid of a good cluster, then none of $\bc{x_, \ldots, x_k, z_1}$ is in a good cluster. In this case, $\pi_{\eta, S}^d(S) = \pi_{\eta, S_1}^d(S)$. 

    \paragraph{Case II: There exists some point $x\in\bc{x_1, \ldots, x_k}$ in a good cluster centered at $c\in S$, \ie $x\in B(c, \eta; S)$ and $B(c, \eta; S)$ is a good cluster.} We first consider the effect of $\pi_{\eta, S}^d$ and $\pi_{\eta, S_1}^d$ on the single point $x$. Note that \begin{equation}
        \norm{z_1 - c}_2 \leq \norm{z_1 - x}_2 + \norm{x - c}_2 \leq 2\eta. 
    \end{equation}
    If $\norm{z_1 - c}_2\leq \eta$, $B(c, \eta; S_1) = B(c, \eta; S\cup \bc{z_1})$ is also a good cluster. Then, $\pi_{\eta, S_1}^d$ and $\pi_{\eta, S_1}^d$ has the same effect on $x$. 
    
    If $\eta \leq \norm{z_1 - c}_2\leq 2\eta$, then $z_1 \in B(c, 3\eta; S_1)$ but $z_1\notin B(c, \eta; S_1)$. This implies $B(c, 3\eta; S_1) \neq B(c, \eta; S_1)$, and $B(c, \eta; S_1)$ is not a good cluster. Therefore, $\pi_{\eta, S}^d$ removes the point $x$, while $\pi_{\eta, S_1}^d$ does not. In this case, $x$ is a different point between $\pi_{\eta, S}^d(S)$ and $\pi_{\eta, S_1}^d(S)$. 

    We note that there are at most $k$ different points between $\pi_{\eta, S}^d(S)$ and $\pi_{\eta, S_1}^d(S)$, when all points $\bc{x_1, \ldots, x_k}$ are in some good cluster $B(c_i, \eta; S)$ for $c_i\in S$ and $z_1$ is selected such that none of $B(c_i, \eta; S_1)$ remains to be a good cluster. In this case, the difference between $\pi_{\eta, S}^d(S)$ and $\pi_{\eta, S_1}^d(S)$ is $\bc{x_1, \ldots, x_k}\subseteq B(z_1, \eta; S)$. 
    
    Following a similar argument, we can show that the maximum number of different points between $\pi_{\eta, S_2}^d (S)$ and $\pi_{\eta, S}^d(S)$ is $k$ and this maximum set of different points is a subset of $B(z_2, \eta; S_2)$. This concludes the proof for the sensitivity of deduplication. 

    \paragraph{Senstivity analysis of quantization} The analysis of $L_\infty$ sensitivity of quantization is the same as that of deduplication. To get the $L_2$ sensitivity of quantization, we consider two cases. If a point is in a good cluster, then quantization process change this point to the centroid of the cluster. The $L_2$ distance incurred by the pre-processing is upper bounded by $\eta$ by the definition of $\eta$-quantization. If a point is not in a good cluster, it remains unchanged after quantization. Combining the two cases, the $L_2$ sensitivity of quantization is upper bounded by $\eta$. 

\end{proof}

\label{app:imputation-tech}

\subsection{Sensitivity analysis of model-based imputation}\label{app:imputation-tech}

In this section, we first provide a general result on the sensitivity for any model-based imputation method. We then introduce specific imputation methods, including mean imputation (\Cref{prop:imputation-mean-sensitivity} in the main text), median imputation, trimmed mean imputation, and linear regression. We summarize their sensitivity results in~\Cref{tab:imputation-sensitivity}. 

For a given model $f$, the imputation algorithm $\pi_f$ first generates an imputation function $f_S$ by fitting the model $f$ to a dataset $S$. Then, it replaces each missing value in the dataset with the prediction based on the imputation function $f_S$. The $L_2$ and $L_\infty$ sensitivity of model-based imputation is presented in~\Cref{prop:imputation-sensitivity}.

Several widely used models for imputation include mean, median, trimmed mean, and linear regression, described as follows. Mean imputation replaces the missing values in the $j^{\text{th}}$ feature with the empirical mean of the available data for that feature. On the other hand, median imputation replaces each missing value with the median of the non-missing points in the corresponding feature. The trimmed mean estimator with parameter $m$ is an interpolation between the mean and median. It estimates each missing value by computing the mean after removing the $m$ smallest and largest points from the remaining points of the feature that is not missing. The above methods address the missing values at a feature $j$ using data from that feature alone. In contrast, linear regression also employs the information of the other features of the missing point. Specifically, it estimates the missing value by the prediction of a linear regression on the $j^{\text{th}}$ feature using some or all other features in the dataset. 

Many imputation methods, such as median and linear regression, do not have bounded global sensitivity even when the instance space is bounded \citep{Alabi22simplelr}. However, the local sensitivity of these methods are usually bounded on well-behaved datasets. Given a collection of well-behaved datasets $\cL$, the $L_2$ sensitivity over the collection is the upper bound on the local sensitivity of all datasets $S\in \cL$. We exploit this property to find the $L_2$ sensitivity of the aforementioned imputation methods. We introduce additional notation and provide a summary of the $L_2$ sensitivity for different methods in~\Cref{tab:imputation-sensitivity}.

\begin{restatable}[Sensitivity of model-based imputation]{proposition}{imputationSensitivity}\label{prop:imputation-sensitivity}
    For a dataset collection $\cL$, the $L_\infty$ sensitivity of model-based imputation $\pi_f$ over $\cL$ is the maximum number of missing values present in any dataset in $\cL$. Furthermore, the $L_2$-sensitivity of $\pi_f$ over $\cL$ is given by \[\Delta_2(\cL, \pi_f) = \max_{\substack{S', S\in \cL\\ d_{H}(S, S')=1}} \max_{x \in S\cup S'}\sqrt{\sum_{j=1}^d\br{f_S( x_j) - f_{S'}( x_j)}^2},\]
    where $x_j$ denotes the $j^{\text{th}}$ feature of $x$. Specifically, for mean imputation, median imputation, trimmed mean imputation, and regression, we present their sensitivities and corresponding assumptions in~\Cref{tab:imputation-sensitivity}. 
\end{restatable}

\begin{table*}[t]
    \centering
    \begin{tabular}{C{0.08\textwidth}C{0.51\textwidth}C{0.17\textwidth}C{0.18\textwidth}}
        \toprule
        Notation & Meaning & Imputation Model $f$ &  $L_2$-Sensitivity 
        $\Delta_2(\cL, \pi_f)$\\
        \midrule
         $p$ & Maximum number of missing points in any $S\in\mathcal{L}$ & Mean &$\frac{2}{n-p}$\\\midrule
         \\[-1em]
         $x_{(m)}^{\min}$ & Minimum $m^{th}$ ordered statistics of any feature in any $S\in \mathcal{L}$& Median & $x_{\br{(n+1)/2}}^{\max} - x_{\br{n/2}}^{\min}$\\ 
         \\[-0.5em]
         $x_{(m)}^{\max}$ & Maximum $m^{th}$ ordered statistics of any feature in any $S\in \mathcal{L}$ &$m$-Trimmed Mean & $\nicefrac{x_{(m)}^{\max} - x_{(n-m)}^{\min}}{n-2m-p}$\\\midrule
         $\lambda_{\max}$ & Maximum eigenvalue of $X^TX$, 
         for any $S = (X, Y)\in \mathcal{L}$&\multirow{2}{*}{Linear regression}&\multirow{2}{*}{$\frac{\lambda_{\max}^2}{(\lambda_{\max} + 1)\lambda_{\min}^2} + \frac{1}{\lambda_{\min}}$}\\
         $\lambda_{\min}$ & Minimum eigenvalue of $X^TX$, 
         for any $S = (X, Y)\in \mathcal{L}$&&\\
        \bottomrule
    \end{tabular}
    \caption{Notations and $L_2$ sensitivity of different imputation models}
    \label{tab:imputation-sensitivity}
\end{table*}


\begin{proof}
    It is obvious that the $L_\infty$ sensitivity of model-based imputation is upper bounded by the number of entries with missing values in any of the dataset $S\in \cL$. The $L_2$ sensitivity is upper bounded by the sensitivity of the model $f$ over the set $\cL$. This concludes the proof.

\textbf{$L_2$ Sensitivity of mean over $\cL$} For neighboring datasets $S_1, S_2\in \cL$, write $S_1 = S\cup \{z_1\}$ and $S_2 = S \cup \{z_2\}$. Without loss of generality, denote $S = \{s_{i}\}_{i = 1}^{n-1}$. For each data point $s_{i}$, we denote its $j^{\text{th}}$ feature as $s_{ij}$. Also, we denote the number of available data points for the $j^{\text{th}}$ feature as $n_j$. In the following, we derive an upper bound on the $L_2$ sensitivity of mean imputation $\Delta_2(\cL, \pi_{\mathrm{mean}})$,

\begin{equation}\label{eq:mean-impu-l2-sensitivity}
\begin{aligned}
        \Delta_2(\cL, \pi_{\mathrm{mean}})^2 &\leq \sum_{j = 1}^d \br{\frac{1}{n_j}\br{\sum_{i = 1}^{n-1}s_{ij} +z_{1j}} - \frac{1}{n_j}\br{\sum_{i = 1}^{n-1}s_{ij} + z_{2j}} }^2\\
        &\leq \sum_{j = 1}^d \br{\frac{1}{n_j}\br{z_{1j}-z_{2j}}}^2\\
        &\leq \frac{1}{n-p}\sum_{j = 1}^d\br{z_{1j}-z_{2j}}^2 = \frac{4}{(n-p)^2}
\end{aligned}
\end{equation}
where the last inequality follows by the fact that the instance space is bounded with diameter $1$. Taking square root from both side of~\Cref{eq:mean-impu-l2-sensitivity} completes the proof. 

\textbf{$L_2$ Sensitivity of median and $m$-trimmed mean over $\cL$} The $L_2$-sensitivity of median and $m$-trimmed mean over $\cL$ follows directly by the definition. 

\textbf{$L_2$ Sensitivity of linear regression $\pi_{\mathrm{LR}}$ over $\cL$} For linear regression, we present the $L_2$ sensitivity by considering two datasets $X, X'\in \cL$ where $X' = X\cup\{z\}$. For any $j\in [d]$, imputation with linear regression of the feature $X_j$ looks at the a submatrix of $X$ that does not include the $j^{th}$ feature. Denote the submatrix that is used for linear regression as $X_r$ ($X_r$ includes a subset of features, specified by the index $r$, of the original dataset $X$). Let $\Sigma = \frac{1}{n}X^\top X$, and let $\Sigma_r = \frac{1}{n}X_r^\top X_r$ be a principal matrix of $\Sigma$ with submatrix $X_r$. We calculate the $L_2$ sensitivity of linear regression on imputing the point $x_j$ below, 

\begin{equation}
\begin{aligned}
        \Delta_2(\cL, \pi_{\mathrm{LR}}) &= \norm{x_r^\top (X_r^\top X_r)^{-1}X_r^\top X_j - x_r^\top (X_r^\top X_r + zz^\top )^{-1}(X_r^\top X_j + z z_j)} \\
        &\leq \norm{x_r}\norm{(X_r^\top X_r)^{-1}X_r^\top X_j -  (X_r^\top X_r + zz^\top )^{-1}(X_r^\top X_j + z z_j)}\\
        &\leq \norm{\br{(X_r^\top X_r)^{-1} - (X_r^\top X_r + zz^\top)}^{-1}X_r^\top X_j - (X_r^\top X_r + zz^\top)^{-1}zz_j} \\
        &\leq \norm{\br{(X_r^\top X_r)^{-1} - (X_r^\top X_r + zz^\top)}^{-1}X_r^\top X_j}+\norm{ (X_r^\top X_r + zz^\top)^{-1}zz_j} \\
        &\overset{(a)}{\leq} \norm{\frac{(X_r^\top X_r)^{-1}zz^\top (X_r^\top X_r)^{-1}}{\br{1+z^\top (X_r^\top X_r)^{-1}z}}X_r^\top X_j } + \frac{1}{\lambda_{\min}}\\
        & \overset{(b)}{\leq} \frac{\lambda_{\max}^2}{(\lambda_{\max} + 1)\lambda_{\min}^2} + \frac{1}{\lambda_{\min}}
\end{aligned}
\end{equation}
where step (a) follows from Sherman-Morrison-Formula and that the eigenvalue of a principal submatrix is always larger than the smallest eigenvalue of the original matrix following Theorem 4.3.15 in \cite{horn_johnson_1985} and then taking supremum over all $X\in \cL$. Similarly, step (b) follows from the fact that the eigenvalue of a principal submatrix is always smaller than the largest eigenvalue of the original matrix following Theorem 4.3.15 in \cite{horn_johnson_1985} and then taking supremum over all $X\in \cL$. 
\end{proof}

\subsection{Sensitivity analysis of PCA}
\PCASensitivity*
\begin{proof}
\textbf{PCA for dimension reduction: }For any two neighboring datasets $S, \tilde{S}$, without loss of generality, we denote $S = \{x_i\}_{i = 1}^n$, and  $\tilde{S} = \{x_i\}_{i = 1}^{n-1}\cup \{x_n'\}$. Let $\Sigma, \tilde{\Sigma}$ denote their empirical covariance matrices and $\hat{\mu}, \hat{\mu}'$ denote the empirical mean. Let $A_k$ and $\tilde{A}_k$ be the matrix consisting of first $k$ eigenvectors of $\Sigma$ and $\tilde{\Sigma}$ respectively. 

First, for any $x\in \cX$, we upper bound the $\norm{A_k^\top x - \tilde{A}_k^\top x}_F$ by a linear function of $\norm{\tilde{\Sigma} - \Sigma}_F$ using properties of the dataset collection $\cL$,
\begin{equation}\label{eq:pca-divergence-bound-l2-sensitivity1}
    \norm{A_k^\top x - \tilde{A}_k^\top x}_F \overset{(a)}{\leq} \norm{A_k - \tilde{A}_k^\top }_F
    \overset{(b)}{=} \text{Tr}\br{2(I- A_k^\top \tilde{A}_k)}
\end{equation}
where step (a) follows from Cauchy-Schwarz inequality and bounded instance space and  step (b) follows from the definition of Frobenius norm $\norm{A}_F = \text{Tr}(A^\top A)$ and the orthonormality of $A_k, \tilde{A}_k$. 

Let $\sigma_i$ denote the $i^{\text{th}}$ singular value of $A_k^\top \tilde{A}_k$, and let $\theta_i$ be the $i^{\text{th}}$ canonical angle of $A_k^\top \tilde{A}_k$,~\ie~ $\cos \theta_i = \sigma_i$. We can write $\text{Tr}\br{2(I- A_k^\top \tilde{A}_k)}$ as follows, 
\begin{equation}\label{eq:pca-divergence-bound-l2-sensitivity}
\begin{aligned}
    \text{Tr}\br{2(I- A_k^\top \tilde{A}_k)}&\overset{(a)}{=} 2\br{k - \sum_{i = 1}^k \sigma_i^2} \overset{(b)}{=} 2\br{k - \sum_{i = 1}^k (\cos\theta_i)^2}\\
    &\overset{(c)}{=}2\br{k - k + \sum_{i = 1}^k (\sin\theta_i)^2} \overset{(d)}{\leq} \frac{4\norm{\tilde{\Sigma}-\Sigma}_F}{\min\{\delta_{\min}^k(\cL), \delta_{\min}^1(\cL)\}}. 
\end{aligned}
\end{equation}
where step (a) and (b) are due to the definition of singular value and canonical angle, step (c) follows from the fact that $(\sin \theta)^2 + (\cos\theta)^2 = 1$ for any $\theta$, and step (d) follows by Davis-Kahan Theorem (\Cref{lem:Davis-Kahan-Theorem}) stated below. 
\begin{lemL}[Davis-Kahan Theorem \citep{yu14daviskahan}]\label{lem:Davis-Kahan-Theorem}
Let $\Sigma, \hat{\Sigma}\in \reals^{d\times d}$ be symmetric and positive definite, with eigenvalues $\lambda_1, \geq \ldots \geq \lambda_d$ and $\hat{\lambda}_1\geq \ldots\geq \hat{\lambda}_d$ respectively. For $k\leq d$, let $V$ and $\hat{V}$ be the dataset whose matrices consisting of the first $k$ eigenvectors of $\Sigma$ and $\hat{\Sigma}$ respectively. Then, 
\[\norm{\sin\Theta(V, \hat{V})}_F\leq \frac{2\min \norm{\Sigma - \hat{\Sigma}}_F}{\min\{\lambda_k - \lambda_{k+1}, \lambda_1-\lambda_2\}}, \]
where $\Theta(V, \hat{V})$ denotes the $k\times k$ diagonal matrix of the principal angles between two subspaces $V$ and $\hat{V}$.  
\end{lemL}
    
    It remains to upper bound the Frobenius norm of $\norm{\Sigma - \tilde{\Sigma}}_F$. By the definition of the empirical covariance matrix, we decompose $\tilde{\Sigma}$ as 
    \begin{equation}\label{eq:pca-sensitivity-0}
        \begin{aligned}
            \tilde{\Sigma} =& \frac{1}{n-1}\sum_{i = 1}^{n-1}(x_i - \hat{\mu}')(x_i - \hat{\mu}')^\top + \frac{1}{n-1}(x_n' - \hat{\mu}')(x_n' - \hat{\mu}')^\top\\
            \overset{(a)}{=}& \underbrace{\frac{1}{n-1}\sum_{i = 1}^{n-1}(x_i - \mu - \frac{1}{n}\br{x_n' - x_n})(x_i - \mu - \frac{1}{n}\br{x_n' - x_n})^\top}_{\text{part I}} \\
            &+ \underbrace{\frac{1}{n-1}(x_n' - \mu - \frac{1}{n}\br{x_n' - x_n})(x_n' - \mu - \frac{1}{n}\br{x_n' - x_n})^\top}_{\text{part II}}
        \end{aligned}
    \end{equation}
    where $(a)$ follows from $\mu' = \mu + \frac{1}{n}\br{x_n' - x_n}$. 
    
    Part (I) can be written as \begin{equation}
        \begin{aligned}\label{eq:pca-sensitivity-part-i}
            \mathrm{Part~I }=&\frac{1}{n-1}\sum_{i = 1}^{n-1}(x_i - \mu - \frac{1}{n}\br{x_n' - x_n})(x_i - \mu - \frac{1}{n}\br{x_n' - x_n})^\top \\
            =&\frac{1}{n-1}\sum_{i = 1}^{n-1}(x_i - \mu) (x_i - \mu)^\top - \frac{1}{n(n-1)}\sum_{i = 1}^{n-1}\br{(x_i - \mu)(x_n' - x_n)^\top + (x_n' - x_n)(x_i - \mu)^\top}\\
            &\qquad+ \frac{1}{n^2}(x_n'-x_n)(x_n'-x_n)^\top \\
            =& \Sigma - \frac{1}{n-1}(x_n-\mu)(x_n-\mu)^\top  + \frac{1}{n(n-1)}\br{(x_n - \mu)(x_n' - x_n)^\top + (x_n' - x_n)(x_n - \mu)^\top} \\
            &\qquad+ \frac{1}{n^2}(x_n' - x_n)(x_n' - x_n)^\top 
        \end{aligned}
    \end{equation}
    Similarly, we can write part (II) as 
    \begin{equation}\label{eq:pca-sensitivity-part-ii}
        \begin{aligned}
            \mathrm{Part~II} =&\frac{1}{n-1}(x_n' - \mu - \frac{1}{n}\br{x_n' - x_n})(x_n' - \mu - \frac{1}{n}\br{x_n' - x_n})^\top\\
            =& \frac{(x_n' - \mu)(x_n' - \mu)^\top}{n-1}  - \frac{(x_n' - \mu) (x_n' - x_n)^\top }{n(n-1)}- \frac{(x_n' - x_n) (x_n' - \mu)^\top}{n(n-1)}\\
            &\qquad+ \frac{(x_n' - x_n)(x_n' - x_n)^\top }{(n-1)n^2}
        \end{aligned}
    \end{equation}
    Substituting~\Cref{eq:pca-sensitivity-part-i} and~\Cref{eq:pca-sensitivity-part-ii} into~\Cref{eq:pca-sensitivity-0}, and by the fact that $\norm{xx^T}_F\leq 1$ for $\norm{x}_2 \leq 1$, we get
    \begin{equation}\label{eq:pca-bounded-difference-in-covariance}
       \begin{aligned}
            \norm{\tilde{\Sigma} - \Sigma}_F \leq \frac{2(3n+2)}{n(n-1)}.
       \end{aligned}
    \end{equation}

    Finally, substituting~\Cref{eq:pca-bounded-difference-in-covariance} into~\Cref{eq:pca-divergence-bound-l2-sensitivity}, and then~\Cref{eq:pca-divergence-bound-l2-sensitivity} into~\Cref{eq:pca-divergence-bound-l2-sensitivity1} completes the proof. 
    
    \textbf{PCA for rank reduction: } For any two neighboring datasets $S, S'\in \cL$, we define the notations of $\tilde{\Sigma}, \Sigma, A_k, \tilde{A}_k$ similarly as in the proof for $\pi_{\mathrm{PCA-dim}}$. We state~\Cref{lem:pca-convergence-projection-space}, which is used to upper bound $\norm{A_kA_k^\top x - \tilde{A}_k\tilde{A}_k^\top x}_F$. 
    \begin{lemL}[Simplified version of Theorem 3 in \cite{ZwaldB05}]\label{lem:pca-convergence-projection-space}
    Let $A$ be a symmetric positive definite matrix with eigenvalues $\lambda_1 > \lambda_2 > \ldots > \lambda_d$. Let $B$ be a symmetric positive matrix. For an integer $k > 0$, let $A_k$ be the matrix consisting the first $k$ eigenvectors of $A$ and $\tilde{A}_k$ be the matrix consisting of the first $k$ eigenvectors of $A + B$. Then, $A_k$ and $\tilde{A}_k$ satisfy that \[\norm{A_k A_k^\top - \tilde{A}_k\tilde{A}_k^\top}\leq\frac{2\norm{B}}{\lambda_k - \lambda_{k+1}}. \]
    \end{lemL}
    Applying~\Cref{lem:pca-convergence-projection-space} with $A+B = \tilde{\Sigma}$ and $A = \Sigma$, we can show an upper bound on the term $\norm{A_kA_k^\top x - \tilde{A}_k\tilde{A}_k^\top x}_F$ for any $x\in \cX$.
    \begin{equation}
        \norm{A_kA_k^\top x - \tilde{A}_k\tilde{A}_k^\top x}_F \leq \frac{\norm{\tilde{\Sigma}-\Sigma}_F}{\lambda_k(S) - \lambda_{k+1}(S)} \leq \frac{4(3n+2)}{n(n-1)\delta_{\min}(S)}
    \end{equation}
    where the last inequality follows by~\Cref{eq:pca-bounded-difference-in-covariance} and the definition of $\delta_{\min}(S) = \min\{\delta_{\min}^k(S), \delta_{\min}^1(S)\}$. Taking the supremum over all dataset $S\in \cL$ concludes the proof. 
\end{proof}

\subsection{Sensitivity Analysis of Scaling}

\StandardScalingSensitivity*

\begin{proof}
The $L_\infty$ sensitivity for both scaling methods is trivially upper bounded by the size of the datasets in $\cL$. In the following, we prove the $L_2$ sensitivity for standard scaling and min max scaling respectively. 

\paragraph{Proof of $L_2$ sensitivity for standard scaling}
    For any two neighboring datasets $S, S'$, let $\mu, \mu'$ denote the mean of $S, S'$ respectively and $\sigma, \sigma'$ denote the standard deviation of $S, S'$ respectively. Then, the $L_2$ sensitivity of standard scaling is 
    \begin{equation}\label{eq:standard-scaling-sensitivity1}
    \begin{aligned}
        \Delta_2 &= \max_{S, S', x}\norm{\frac{x-\mu}{\sigma} - \frac{x - \mu'}{\sigma'}}_2\\
        &= \max_{S, S', x} \norm{\frac{\sigma'x - \sigma' \mu - \sigma x + \sigma \mu'}{\sigma \sigma'}}\\
        &= \max_{S, S', x}\norm{\frac{\br{\sigma' - \sigma}(x - \mu)}{\sigma\sigma'}}_2 + \max_{S, S'}\norm{\frac{\mu - \mu'}{\sigma'}}\\
        &\overset{(a)}{\leq} \frac{2\max_{S, S'}\norm{\sigma' - \sigma}}{\sigma_{\min}^2} + \frac{\max_{S, S'}\norm{\mu - \mu'}}{\sigma_{\min}}
    \end{aligned}
    \end{equation}
    where step (a) follows from $\norm{x - \mu}_2\leq 2$ for any $x$ and $ S$ and the definition of $\sigma_{\min}$. 

    From~\citet{liu2016globalsensitivity}, the global sensitivity of sample mean and variance are $\frac{2}{n}$. We then show that for a dataset collection $\cL$, the global sensitivity conditional on $\cL$ is upper bounded by $\nicefrac{1}{\sigma_{\min}n }$. 
    
    For any $S, S'$ with sample variances $\sigma^2, \br{\sigma'}^2$,  
    we have \begin{equation}\label{eq:std-sample-sensitivity}
        \abs{\sigma - \sigma'} \leq \frac{\abs{\sigma^2-\br{\sigma'}^2}}{2\sigma_{\min}} \leq \frac{1}{\sigma_{\min}n},
    \end{equation}
    where the first inequality follows by rearranging $\abs{\sigma^2 - \br{\sigma'}^2 }= \abs{\br{\sigma - \sigma'}\br{\sigma + \sigma'}} \geq 2\sigma_{\min} \abs{\sigma - \sigma'}$, and the second inequality by substituting the global sensitivity of sample covariance from~\citet{liu2016globalsensitivity}. 

    Substituting the sensitivity of sample mean $2/n$ and sample standard deviation $\frac{1}{\sigma_{\min}n}$ into~\Cref{eq:standard-scaling-sensitivity1}, we conclude the proof. 
\end{proof}

\clearpage

\section{Proofs for overall privacy guarantees (\Cref{sec:overall-privacy})}
\label{app:sec4-proofs}

In this section, we provide the proofs of the overall privacy guarantees for specific pre-processed DP pipelines, as stated in~\Cref{tab:comparison} in~\Cref{sec:overall-privacy}. For clarity, we restate the full version of~\Cref{thm:overall-privacy-table2} and specify the privacy guarantee for each category of privacy mechanisms (each row in~\Cref{tab:comparison}) in~\Cref{thm:overall-privacy}. Then, we present the proof for each category separately. As the proof of DP-GD is the same as that of Gaussian mechanism, we omit the proof for DP-GD for simplicity. 

\begin{thm}[Full version of~\Cref{thm:overall-privacy-table2}]
    \label{thm:overall-privacy}
    Let $p$ denote the $L_\infty$ sensitivity of deduplication, quantization and mean imputation. Let $n \geq \min\bc{101, p}$ be the size of any dataset in the dataset collection $\cL$. Let $\ell$ be a $1$-Lipschitz and $1$-smooth loss function. For an output function $f$ and a score function $Q$, assume their Lipschitz parameter and global sensitivity are both $1$. Then, 
    \begin{itemize}
        \item[(i)] Gaussian mechanism with output function $f$ satisfy $\br{\alpha, 
        {1.05\alpha\varepsilon^2}(1+p^2)}$-RDP, $\br{\alpha, {1.05\alpha\varepsilon^2}(1+\eta^2p^2)}$-RDP, $\br{\alpha, {1.05\alpha\varepsilon^2}\br{1+\frac{4p^2}{(n-p)^2}}}$-RDP, $\br{\alpha, {1.05\alpha\varepsilon^2}\br{1+\frac{12.2^2}{\delta_{\min}^2}}}$-RDP and $1.05\alpha\varepsilon^2\br{1 + \frac{4}{\sigma_{\min}^3}}$-RDP when coupled with deduplication, quantization, mean imputation, PCA for rank reduction and standard scaling respectively. 
        \item[(ii)] Exponential mechanism with score function $Q$ and Laplace mechanism with output function $f$ satisfy $(\alpha, \varepsilon(1+p))$-RDP, $(\alpha, \varepsilon(1+\eta p))$-RDP, $\br{\alpha, \varepsilon\br{1+\frac{2p}{n-p}}}$-RDP, $\br{\alpha, \varepsilon\br{1+\frac{12.2}{\delta{\min}}}}$-RDP and $4.2\alpha\varepsilon^2\br{1 + \frac{1}{\sigma_{\min}^3}}$-RDP when coupled with deduplication, quantization, mean imputation, PCA and standard scaling for rank reduction respectively. 
        \item[(iii)] DP-SGD with subsampling with loss function $\ell$ satisfies $\br{\alpha, {1.05\alpha\varepsilon^2}\br{2\alpha + \frac{12.2^2}{\delta_{\min}^2}}}$-RDP and $2.1\alpha\varepsilon^2\br{\alpha + \frac{8}{\sigma_{\min}^6}}$-RDP when coupled with PCA for rank reduction and standard scaling respectively. 
        \item[(iv)] DP-SGD with iteration with loss function $\ell$ coupled with deduplication, quantization and mean imputation satisfy $\br{\alpha, {1.1\alpha\varepsilon^2}\br{1 + \frac{p^2n \log(n-p)}{(n-p)\log n}}}$-RDP, $\br{\alpha, {1.1\alpha\varepsilon^2}\br{1 + \frac{\eta^2p^2n \log(n-p)}{(n-p)\log n}}}$-RDP, and\\ $\br{\alpha, {1.1\alpha\varepsilon^2}\br{1 + \frac{4p^2n \log(n-p)}{(n-p)^3\log n}}}$-RDP respectively.
    \end{itemize}

\end{thm}
\begin{proof}[Proof of~\Cref{thm:overall-privacy} (i)]
    We apply~\Cref{thm:general-thm-training-only} with $c_1 = c_2 = 2$, 
    \begin{equation}\label{eq:overall-priv-gaussian}
    \begin{aligned} 
        \widehat{\varepsilon}(\alpha)&\leq \max\left\{\frac{2\alpha - 1}{2(\alpha - 1)}\tilde{\varepsilon}(2\alpha ,\Delta_2\Delta_2) + \varepsilon(2\alpha - 1), \frac{2\alpha - 1}{2(\alpha - 1)}\varepsilon(2\alpha ) + \tilde{\varepsilon}(2\alpha - 1, \Delta_2\Delta_\infty)\right\}\\
        &\overset{(a)}{\leq} \max\left\{\frac{2\alpha - 1}{2(\alpha - 1)}\tilde{\varepsilon}(2\alpha ,\Delta_2\Delta_\infty) + \frac{2\alpha - 1}{2(\alpha - 1)}\varepsilon(2\alpha), \frac{2\alpha - 1}{2(\alpha - 1)}\varepsilon(2\alpha ) + \frac{2\alpha - 1}{2(\alpha - 1)}\tilde{\varepsilon}(2\alpha, \Delta_2\Delta_\infty)\right\}\\
        &= \frac{2\alpha - 1}{2(\alpha - 1)}\br{\tilde{\varepsilon}(2\alpha ,\Delta_2\Delta_2) + \varepsilon(2\alpha)}\\
        &\overset{(b)}{\leq}1.05\br{\tilde{\varepsilon}(2\alpha ,\Delta_2\Delta_\infty) + \varepsilon(2\alpha)}. 
    \end{aligned}
    \end{equation}
    where step (a) follows from the monotonicity of Renyi Divergence (\Cref{lem:monotonicity-RDP}) and step (b) follows from $\frac{2\alpha - 1}{2(\alpha - 1)} \leq1.05$ for $\alpha > 11$. 

    By substituting the expression of RDP and SRDP parameter from~\Cref{tab:sdp-parameters} for Gaussian mechanism, the $L_2$ and $L_\infty$ sensitivity ($\Delta_2$ and $\Delta_\infty$) for deduplication, quantization, PCA, standard scaling and mean imputation from \Cref{prop:deduplication-sensitivities,prop:imputation-mean-sensitivity,prop:pca-sensitivity,prop:sensitivity-scaling}, and the Lipschitz parameter and global sensitivity of the output function $f$ into~\Cref{eq:overall-priv-gaussian}, we complete the proof of the overall privacy guarantees for Gaussian mechanism.

  
\end{proof}

\begin{proof}[Proof of~\Cref{thm:overall-privacy} (ii)]
We apply~\Cref{thm:general-thm-training-only} with $c_1 = c_2= 1$, 
    \begin{equation}\label{eq:overall-priv-laplace}
    \begin{aligned} 
        \widehat{\varepsilon}(\alpha)&\leq \max\left\{\frac{\alpha - 1}{\alpha - 1}\tilde{\varepsilon}(\alpha ,\Delta_2\Delta_2) + \varepsilon(\infty), \frac{\alpha - 1}{\alpha - 1}\varepsilon(\alpha ) + \tilde{\varepsilon}(\infty, \Delta_2\Delta_\infty)\right\}\\
        &= \max\{\tilde{\varepsilon}(\alpha, \Delta_2\Delta_\infty) + \varepsilon(\infty), \varepsilon(\alpha) + \tilde{\varepsilon}(\infty, \Delta_2\Delta_\infty)\}
    \end{aligned}
    \end{equation}
    We then derive an upper bound on $\widehat{\varepsilon}(\alpha)$ by monotonicity of RDP, 
    \begin{equation}\label{eq:overall-priv-laplace2}
        \widehat{\varepsilon}(\alpha)\overset{(a)}{\leq} \widehat{\varepsilon}(\infty)\overset{(b)}{\leq} \tilde{\varepsilon}(\infty, \Delta_2\Delta_\infty) + \varepsilon(\infty)
    \end{equation}
    where step (a) follows from~\Cref{lem:monotonicity-RDP} and step (b) follows by setting $\alpha = \infty$ in~\Cref{eq:overall-priv-laplace}. 

    By substituting the expression of RDP and SRDP parameter from~\Cref{tab:sdp-parameters} for Laplace mechanism, the $L_2$ and $L_\infty$ sensitivity ($\Delta_2$ and $\Delta_\infty$) for deduplication, quantization, imputation and PCA from \Cref{prop:deduplication-sensitivities},~\Cref{prop:imputation-mean-sensitivity} and~\Cref{prop:pca-sensitivity}, and the Lipschitz parameter and global sensitivity of the output function $f$ into~\Cref{eq:overall-priv-gaussian}, we complete the proof of the overall privacy guarantees for Laplace mechanism.

    Similarly, for exponential mechanism, we substitute the expression of RDP and SRDP parameter from~\Cref{tab:sdp-parameters} for exponential mechanism, the $L_2$ and $L_\infty$ sensitivity ($\Delta_2$ and $\Delta_\infty$) for deduplication, quantization, PCA, standard scaling and mean imputation from \Cref{prop:deduplication-sensitivities,prop:imputation-mean-sensitivity,prop:pca-sensitivity,prop:sensitivity-scaling}, and the Lipschitzness and the global sensitivity of the score function $Q$, into~\Cref{eq:overall-priv-gaussian}. This completes the proof of the overall privacy guarantees for exponential mechanism. 
  
\end{proof}

In the following, we provide the proof of combining DP-SGD with subsampling with PCA for rank reduction because datasets preprocessed by PCA for rank reduction satisfies~\Cref{assump:dp-sgd-subsampling} automatically. 

\begin{proof}[Proof of~\Cref{thm:overall-privacy} (iii)]
    We first note that for PCA, for any two neighboring datasets, $S_1, S_2$, by the definition of sensitivity of PCA, $\pi_{\mathrm{PCA-rank}}(S_1),\pi_{\mathrm{PCA-rank}}(S_2) $ and $\pi_{\mathrm{PCA-dim}}(S_1),\pi_{\mathrm{PCA-dim}}(S_2)$ has inverse point-wise divergence $1$ because $\Delta_\infty = n$. 

    Following a similar argument as in the proof Gaussian mechanism (\Cref{thm:overall-privacy} (i)), we set the parameters in~\Cref{thm:general-thm-training-only} $c_1 = c_2 = 2$ and obtain~\Cref{eq:overall-priv-gaussian}. Then, we substitute in~\Cref{eq:overall-priv-gaussian} the following parameters: a)RDP and SRDP parameters for DP-SGD with subsampling from~\Cref{tab:sdp-parameters}, b)$\gamma = 1$, c) Lipschitzness and smoothness parameter of the loss function $\ell$, and d) $L_2$ and $L_\infty$ sensitivity ($\Delta_2$ and $\Delta_\infty$) for PCA and standard scaling from~\Cref{prop:pca-sensitivity,prop:sensitivity-scaling}. This completes the proof. 
\end{proof}

For DP-SGD with iteration, the pre-processing mechanisms achieve a tighter bound while satisfying $\Delta_\infty\ll n$ by~\Cref{assump:dp-sgd-iteration}. Therefore, this method will not provide tighter bound for PCA and standard scaling. However, it improves the privacy analysis for imputation, deduplication and quantization, as proved below. 

\begin{proof}[Proof of~\Cref{thm:overall-privacy} (iv)]
    We first note that the maximum divergence $\maxD$ of datasets pre-processed by imputation, deduplication or quantization is upper bounded by $\Delta_\infty=p$. 

    Following a similar argument as in the proof of Gaussian mechanism (\Cref{thm:overall-privacy} (i)), we set the parameters in~\Cref{thm:general-thm-training-only} $c_1 = c_2 = 2$. We substitute in~\Cref{eq:overall-priv-gaussian} the following parameters: a)RDP and SRDP parameters for DP-SGD with iteration from~\Cref{tab:sdp-parameters}, b)$\maxD = \Delta_\infty = p$, c) Lipschitzness and smoothness parameter of the loss function $\ell$, and d) $L_2$ and $L_\infty$ sensitivity ($\Delta_2$ and $\Delta_\infty$) for quantization, deduplication and imputation from from~\Cref{prop:deduplication-sensitivities} and \Cref{prop:imputation-mean-sensitivity}. This completes the proof. 
\end{proof}

\clearpage

\section{Privacy and accuracy guarantees of~\Cref{alg:ptr-imputation} (\Cref{sec:ptr})}\label{app:sec5-proofs}

In this section, we present the proofs for the theoretical guarantees of \Cref{alg:ptr-imputation}, as stated in~\Cref{sec:ptr}. Specifically, we provide the proofs for the privacy guarantee (\Cref{thm:ptr-privacy-guarantee}) and the accuracy guarantee (\Cref{prop:pca-accuracy}) of~\Cref{alg:ptr-imputation}.

\PTRGuarantee*
\begin{proof}
For brevity, we denote~\Cref{alg:ptr-imputation} as $\PTR$. For any two neighboring datasets $S_1, S_2\in \cL$, we consider two cases: i) the minimum eigen-gap of either $S_1$ or $S_2$ is smaller than or equal to $\beta$, and ii) the minimum eigen-gap of both $S_1$ and $S_2$ is greater than $\beta$. 

    \textbf{Case (i)} Without loss of generality, we assume the eigen-gap of $S_1$ is smaller than or equal to $\beta$. We will show that for any neighboring dataset $S_2$ of $S_1$, \ie~$d_H(S_1, S_2) = 1$, $\PTR$ satisfies the following inequality, for any output set $O\subset \cH\cup \{\perp\}$
    \begin{equation}
        \label{eq:ptr-case-i-goal}
        \bP\bs{\PTR(S_1) \in O} \leq e^\epsilon \bP\bs{\PTR(S_2)\in O} + \frac{\delta}{2}.
    \end{equation}
    First consider the case  $O\subset \cH$. Then, assuming that the output is $\perp$ with high probability~\ie~$\bP\bs{\PTR(S_1)\neq \perp} \leq \nicefrac{\delta}{2}$, we have that \begin{equation}
        \label{eq:ptr-case-i-1}
        \bP[\PTR(S_1)\in O] \leq \bP[\PTR(S_1)\neq \perp] \leq \frac{\delta}{2} \leq e^{\epsilon}\bP[\PTR(S_2) \in O] + \frac{\delta}{2}.
    \end{equation}
    Now we show $\bP\bs{\PTR(S_1)\neq \perp} \leq \nicefrac{\delta}{2}$. Given a input dataset $S$, we denote the $\Gamma$ in step 1 of~\Cref{alg:ptr-imputation} by $\Gamma(S)$. 
    
    \begin{equation}\label{eq:ptr-case-1-1}
        \begin{aligned}
            \bP(\PTR(S_1) \neq \perp) &\overset{(a)}{\leq} \bP\br{\Gamma(S_1) \leq \frac{\log\frac{2}{\delta}}{\varepsilon}}
            &\overset{(b)}{=} \bP\br{\Lap{\frac{1}{\varepsilon}}\leq \frac{\log\frac{2}{\delta}}{\varepsilon}}\overset{(c)}{\leq} \frac{\delta}{2}.
        \end{aligned}
    \end{equation}
    where step (a) follows because the algorithm does not return $\perp$ only if $\Gamma(S_1) \geq \frac{\log\frac{2}{\delta}}{\varepsilon}$(Line 5 in~\Cref{alg:ptr-imputation}), and step (b) follows from $\min_{S':\delta_{\min}(S')\leq \beta}d_H(S', S_1)=0$ as $S_1$ itself satisfies $\delta_{\min}(S_1) \leq \beta$. Step $(c)$ follows by the tail bound of Laplace distributions: for a positive real number $t > 0$, $\bP_{z\sim \mathrm{Lap}(0, b)}(z\geq t) \leq e^{-\nicefrac{t}{b}}$. 

    Now consider the only other possible output $O=\bc{\perp}$. We note that for any neighboring dataset $S_2$ of $S_1$, 
    \begin{equation}\label{eq:ptr-case-i-2}
        \min_{S':\delta_{\min}(S')\leq \beta}d_H(S', S_1)=0, \quad \min_{S':\delta_{\min}(S')\leq \beta}d_H(S', S_2)\leq 1
    \end{equation}
This implies 
    \begin{equation}
        \label{eq:ptr-case-i-1}
        \bP\br{\Gamma(S_2)\leq \frac{\log \frac{2}{\delta}}{\varepsilon}} \geq \bP\br{\Lap{\frac{1}{\varepsilon}}\leq \frac{\log \frac{2}{\delta}}{\varepsilon}-1}.
    \end{equation}   
    Define $J(S) = \min_{S':\delta_{\min}(S')< \beta}d_H(S, S')$. It is simple to note that $J$ has global sensitivity $1$ due to~\Cref{eq:ptr-case-i-2}. Thus, $\Gamma$ is essentially Laplace mechanism on \(J\) and thereby $(\varepsilon, 0)$-DP. This implies that for any neighboring datasets $S_2$ of $S_1$,
    \begin{equation}\label{eq:ptr-case-1-2}
            \frac{\bP(\PTR(S_1) = \perp)}{\bP(\PTR(S_2) = \perp)}
            = \frac{\bP\br{\Gamma(S_1) \leq \frac{\log\frac{2}{\delta}}{\varepsilon}}}{\bP\br{\Gamma(S_2) \leq \frac{\log\frac{2}{\delta}}{\varepsilon}}} \overset{(a)}{\leq} \frac{\bP\br{\Lap{\frac{1}{\varepsilon}} \leq \frac{\log\frac{2}{\delta}}{\varepsilon}}}{\bP\br{\Lap{\frac{1}{\varepsilon}} \leq \frac{\log\frac{2}{\delta}}{\varepsilon}-1}} \overset{(b)}{\leq} e^\varepsilon
    \end{equation}
    where the numerator of step (a) follows by~\Cref{eq:ptr-case-1-1} and the denominator of step (a) follows by~\Cref{eq:ptr-case-i-1}. Step (b) follows by the tail bound of Laplace distributions. Combining~\Cref{eq:ptr-case-1-1} and~\Cref{eq:ptr-case-1-2}, we show that~\Cref{eq:ptr-case-i-goal} holds in case (i). 

    \textbf{Case (ii) } Consider $S_1, S_2$ whose $k^{\text{th}}$-eigengaps are greater than $\beta$. We will show that the following holds, for any $O\subset \cH\cup\{\perp\}$
    \begin{equation}
        \label{eq:ptr-case-ii-goal}
        \bP\bs{\PTR(S_1)\in O} \leq e^{\widehat{\varepsilon}}\bP\bs{\PTR(S_2)\in O} + \frac{\delta}{2}
    \end{equation}
    where $ \widehat{\varepsilon} = 3\varepsilon\sqrt{1.05\br{1 + \frac{12.2^2\mu^2}{L^2\beta^2}}\log\frac{1}{\delta}}$.
    First consider the case $O = \{\perp\}$. Following a similar argument as~\Cref{eq:ptr-case-1-2}, we have \begin{equation}
        \label{eq:ptr-case-2-1}
        \frac{\bP\bs{\PTR(S_1) = \perp}}{\bP\bs{\PTR(S_2) = \perp}}\leq e^\varepsilon \leq e^{\widehat{\varepsilon}}.
    \end{equation}  
        Now consider the case when the output is not $\perp$. Then,  choosing any $\alpha \geq 11$ yields that the privacy parameter of DP-GD combined with non-private PCA is $\br{\alpha, 1.05\alpha\varepsilon\br{1 + \frac{12.2^2\mu^2}{L^2\beta^2}}}$-RDP from~\Cref{tab:comparison}. Then, invoking~\Cref{lem:rdp-to-adp} converts the RDP parameter to Approximate DP.
\begin{lemL}[RDP to Approximate DP \citep{mironov2017renyi}]\label{lem:rdp-to-adp}
    If $\cA$ is an $(\alpha, \varepsilon)$-RDP algorithm, then for $0<\delta < 1$, it satisfies $\br{\varepsilon + \frac{\log \frac{1}{\delta}}{\alpha-1}, \delta}$-differential privacy. 
\end{lemL}
        In particular, we choose \(\alpha = \sqrt{\frac{\log\frac{1}{\delta}}{1.05\varepsilon^2\br{1 + \frac{12.2^2\mu^2}{L^2\beta^2}}}} + 1\geq 11\) and obtain that the output in step 6 of~\Cref{alg:ptr-imputation} is $(\widehat{\varepsilon}, \delta)$-DP, where 
        \[\widehat{\varepsilon} \overset{(a)}{\leq} 2\varepsilon \sqrt{1.05\br{1 + \frac{12.2^2\mu^2}{L^2\beta^2}}\log\frac{1}{\delta}} + 1.05\varepsilon\br{1 + \frac{12.2^2\mu^2}{L^2\beta^2}}\overset{(b)}{\leq} 3\varepsilon\sqrt{1.05\br{1 + \frac{12.2^2\mu^2}{L^2\beta^2}}\log\frac{1}{\delta}} \]
       Here, step (a) follows by~\Cref{lem:rdp-to-adp} with chosen $\alpha$ and step (b) utilizes the fact that $\delta$ satisfies $\log\frac{1}{\delta}\geq1.05\varepsilon^2\br{1 + \frac{12.2^2\mu^2}{L^2\beta^2}}$. 
        
        This algorithm also discloses information about the dataset regarding its minimum eigen-gap. Specifically, when the minimum eigen-gap of the private dataset is smaller than $\beta$, then with high probability the output is $\perp$. However, the additional privacy cost incurred is smaller than that of releasing $\Gamma$ directly. As $\Gamma$ is a Laplace mechanism with global sensitivity $1$, releasing $\Gamma$ satisfies $(\varepsilon, 0)$-DP. By basic composition of approximate differential privacy~(\Cref{lem:composition-DP}), we can combine the privacy cost for releasing the results from DP-GD and releasing $\Gamma$. Thus,~\Cref{alg:ptr-imputation} is $\br{\widehat{\varepsilon} + \varepsilon, \delta}$-DP under case (ii).
        \begin{lemL}[Basic composition of differential privacy \cite{dwork14algorithmicfoundation}]
        \label{lem:composition-DP}
        If algorithm $\cA_1$ is $(\varepsilon_1, \delta_1)$-DP and $\cA_2$ is $(\varepsilon_2, \delta_2)$-DP, then $(\cA_1, \cA_2)$ is $(\varepsilon_1 + \varepsilon_2, \delta_1 + \delta_2)$-DP. 
        \end{lemL} 
        Combining the two cases concludes the proof. 
\end{proof}

\PCAaccuracy*
\begin{proof}
    We first prove that the probability that~\Cref{alg:ptr-imputation} outputs \(\perp\) is small. Then, we show that in the even that it does not output \(\perp\) but a real vector, then the excess empirical risk is small.

    Let $\cS_{d_{H}} = \bc{S, S'\mid d_H(S, S') = 1\wedge \delta_{\min}(S), \delta_{\min}(S')>0}$. Let $\Sigma_S, \Sigma_{S'}$ denote the covariance matrices of $S$ and $S'$ respectively and \(\delta_k\br{S}\) denote the \(k^{\it th}\) eigen-gap of \(\Sigma_S\). 
    Then, for any $S, S'\in \cS_{d_H}$, we first upper bound \(\abs{ \delta_k\br{S}- \delta_{k}\br{S'}}\)

 \begin{equation}
        \begin{aligned}
            \abs{\delta_{k}(S') - \delta_{k}(S)}&= \abs{\lambda_k(\tilde{\Sigma}) - \lambda_{k+1}(\tilde{\Sigma}) - \lambda_k(\Sigma) + \lambda_{k+1}(\Sigma)}\\
            &\leq \abs{\lambda_k(\tilde{\Sigma}) - \lambda_k (\Sigma)} + \abs{\lambda_{k+1}(\Sigma) - \lambda_{k+1}(\tilde{\Sigma})}.\\
        \end{aligned}
    \end{equation}
 To bound this term, we upper bound \(\abs{\lambda_k(\tilde{\Sigma}) - \lambda_k (\Sigma)}\) for any \(k\).  Using Weyl's inequality~\citep{Schindler2015Wehl} 
    \begin{equation}\label{eq:eigenvalue-sensitivity}
        \max_{j \in [d]}\abs{\lambda_j(\Sigma) - \lambda_j(\tilde{\Sigma})}\leq \norm{\tilde{\Sigma}-\Sigma}_{\text{op}}\leq\norm{\tilde{\Sigma}-\Sigma}_F\leq \frac{6.1}{n},
    \end{equation}        
    where the last inequality follows from a similar argument as \Cref{eq:pca-bounded-difference-in-covariance} in the proof of~\Cref{prop:pca-sensitivity} for $n\geq 101$. This yields 
\begin{equation}\label{eq:eigen-gap-sensitivity}
     \abs{\delta_{k}(S') - \delta_{k}(S)}\leq \frac{12.2}{n}.
\end{equation}
    Thus, for any $S_0,S_m$ with $d_H(S_0, S_m) = m$, we can construct a series of datasets $S_1, \ldots, S_{m-1}$ such that $d_H(S_j, S_{j+1}) = 1$ for $j\in \{0,...,m-1\}$. By applying~\Cref{eq:eigen-gap-sensitivity} iteratively over each $i$, we have 
    \begin{equation}\label{eq:k-sensitivity-eigen-gap}
        \delta_{k}(S_0)\leq \delta_{k}(S_m) + \frac{12.2m}{n}.
    \end{equation}
     More generally, for any pairs of dataset \(S,S'\), we have  
     \begin{equation}\label{eq:k-sensitivity-eigen-gap}
        d_H(S,S') \geq \frac{\abs{\delta_k\br{S} - \delta_k\br{S'}}n}{12.2}
    \end{equation} Hence,
    \begin{equation}
        \label{eq:pca-acc-ineq}
        \min_{S': \delta_{k}(S')\leq \beta} d_H(S, S') \geq \frac{\br{\delta_{k}(S)-\beta}n}{12.2}.
    \end{equation}
    Now, we upper bound the probability that ~\Cref{alg:ptr-imputation} does not output $\perp$.
    \begin{equation}\label{eq:pca-acc-high-probability-guarantee}
        \begin{aligned}
            \bP\bs{\Gamma\br{S} \geq \frac{\log\frac{1}{\delta}}{\varepsilon}} &= \bP\bs{\min_{S': \delta_{k}(S')\leq \beta}d_H(S, S') + \mathrm{Lap}\br{\frac{1}{\varepsilon}} \geq \frac{\log\frac{1}{\delta}}{\varepsilon}} \\
            &\geq \bP\bs{\frac{\br{\delta_{k}(S)-\beta}n}{12.2} + \mathrm{Lap}\br{\frac{1}{\varepsilon}}\geq \frac{\log\frac{1}{\delta}}{\varepsilon}}\\
            &\overset{(a)}{=} 1-\frac{1}{2}\exp\br{\varepsilon\br{\frac{\log\frac{1}{\delta}}{\varepsilon} - \frac{\br{\delta_{k}(S)-\beta}n}{12.2}}}\\
            &= 1-\frac{1}{2\delta}\exp\br{\frac{\beta - \delta_{k}(S)}{12.2}n\varepsilon},
        \end{aligned}
    \end{equation}
    where step $(a)$ follows from the tail bound of Laplace distribution when $\delta_{k}(S)\geq \beta + \frac{12.2\log\frac{1}{\delta}}{n\varepsilon}$. 
    
    We have shown that~\Cref{alg:ptr-imputation} does not output $\perp$ with probability $1-\frac{1}{2\delta}\exp\br{\frac{\beta - \delta_{k}(S)}{12.2}n\varepsilon}$. It remains to derive the convergence guarantee of~\Cref{alg:ptr-imputation} when the output is not $\perp$. To show this, we utilize the property that PCA transforms the dataset to have a low-rank covariance matrix, which allows us to apply the dimension-independent convergence guarantee of DP-GD following the analysis of \citet{song21Evading} and establish the desired convergence guarantee.

    Let $A_k^\top \in \reals^{k\times d}$ be the matrix consisting of the first $k$ eigenvectors of the covariance matrix of $S$, $\Sigma = \frac{1}{n}\sum_{i = 1}^n x_i x_i^\top $. Let $S_k = \{(A_k A_k^\top x_i, y_i)\}_{i = 1}^n$. We can decompose the error into three terms, 
    \begin{equation}
    \label{eq:pca-acc-decomposition}
        \bE\bs{\hat{\ell}_S(\hat{\theta})} - \ell_S(\theta^\star) \leq \underbrace{ \abs{\bE\bs{\hat{\ell}_S(\hat{\theta})} - \bE\bs{\hat{\ell}_{S_k}(\hat{\theta})} }}_{(a)} + \underbrace{\abs{\bE\bs{\hat{\ell}_{S_k}(\hat{\theta})} -\ell_{S_k}(\theta^\star)}}_{(b)} + \underbrace{\abs{\ell_S(\theta^\star)-\ell_{S_k}(\theta^\star)}}_{(c)}
    \end{equation}
    By Theorem 3.1 in \citet{song21Evading}, part $(b)$ is upper bounded by $\bigO{\frac{L\sqrt{1 + k\log\frac{1}{\delta}}}{\varepsilon n}}$.  
\begin{lemL}[Theorem 3.1 in~\citet{song21Evading}]
    Let $\theta_0 = 0^p$ be the initial point of $\DPGD$. Let the dataset be centered at $0$. Let $k$ be the rank of the projector to the eigenspace of the covariance matrix $\sum_{i = 1}^n x_ix_i^\top$. For a $L$-Lipschitz loss function $\ell$, $\DPGD$ with $T = n^2 \varepsilon^2$ with appropriate learning rate $\eta$ output $\hat{\theta}$ satisfying 
    \[\bE\bs{\hat{\ell}_S(\hat{\theta})} - \hat{\ell}_S(\theta^\star) \leq \bigO{\frac{L\sqrt{1 + 2k\log\frac{1}{\delta}}}{\varepsilon n}}. \]
\end{lemL}
    Next, we show that part $(a)$ and $(c)$ are upper bounded by $L\sum_{i = k+1}^n\lambda_i(S)= L\Lambda$. To show this, we prove that for any $\theta \in B_2^d$, $\abs{\hat{\ell}_{S_k}(\theta) - \hat{\ell}_S(\theta)}\leq L\sum_{i = k+1}^n\lambda_i(S)$. 

    \begin{equation}
        \begin{aligned}
            \abs{\hat{\ell}_{S_k}(\theta) - \hat{\ell}_S(\theta)} &= \abs{\frac{1}{n}\sum_{i = 1}^n \ell(\theta^\top A_k A_k^\top x_i, y_i) - \frac{1}{n}\sum_{i = 1}^n \ell(\theta^\top x_i, y_i)}\\ 
            &\overset{(a)}{\leq} \frac{1}{n}\sum_{i = 1}^n L\norm{\theta^\top A_kA_k^\top x_i - \theta^\top x_i}_2\\
            &\overset{(b)}{\leq} \frac{1}{n}\sum_{i = 1}^n L\norm{ A_kA_k^\top x_i -  x_i}_2 \overset{(c)}{\leq} L\sum_{i = k+1}^d \lambda_i(S) = L\Lambda
        \end{aligned}
    \end{equation}
    where step $(a)$ follows from Lipschitzness of the loss function $\ell$, step $(b)$ follows due to the projection step in DP-GD projects $\theta$ to the Euclidean ball with radius 1, and step $(c)$ is obtained by substituting the reconstruction error of PCA, which is $\sum_{i = k+1}^d \lambda_i(S)$. This concludes the proof. 
\end{proof}
\clearpage

\section{Experiment setups and discussion}\label{app:experiments}
\subsection{Experiment setups}

\paragraph{Data generation} The synthetic data is generated with the \verb|make_classification| function in the \verb|sklearn| library. We generate a 2-class low rank dataset consisting of 1000 data points with dimension 6000 and rank 50. We set the parameter \verb|n_cluster_per_class| in \verb|make_classification| to 1. 

\paragraph{Models and allocation of privacy budget} We compared the excess empirical loss of three models in~\Cref{fig:training-acc-pca}. We provide the details of the models below. 
\begin{itemize}
    \item\textbf{Pre-processed DP pipeline: }We employ non-private PCA to reduce the dimensionality of the original dataset to $k$ and then apply private logistic regression. In particular, we use the \verb|make_private_with_epsilon| method from the \verb|Opacus| library with PyTorch SGD optimizer with learning rate \verb|1e-2|, \verb|max_grad_norm = 10| and \verb|epochs = 10|. The privacy parameters \verb|epsilon| and \verb|delta| are obtained by adjusting the desired overall privacy level with our framework (\Cref{thm:overall-privacy}). 
    \item \textbf{No pre-processing: }We directly apply private logistic regression with the same parameters on the original high-dimensional dataset. 
    \item \textbf{DP-PCA: }We first implement the DP-PCA in \citet{chaudhuri2012dppca} and then apply private logistic regression. We allocate half of the privacy budget to DP-PCA and the remaining half to private logistic regression.
\end{itemize}

\subsection{Discussion on clipping}
In practice, when the sensitivity of the original function is unbounded, one can apply the technique of clipping to restrict the sensitivity~\citep{Abadi16dpsgd,Liu22dp-pca}. We can also incorporate clipping into our pre-processed DP pipeline. Specifically, we first non-privately compute the PCA projection matrix $A_k$ with the original dataset $S$. We clip the original data to $S_{\text{clipped}}$ with the clipping threshold $C \in [0.1R, 0.7R, 0.99R]$, where $R$ represents the maximum norm of the original dataset. Then, we apply the projection matrix to the clipped dataset to obtain the pre-processed dataset $S_{\text{preprocessed}} = A_kA_k^\top S_{\text{clipped}}$. Finally, we apply private logistic regression on $S_{\text{preprocessed}}$, with the privacy parameter set by our framework. 

As shown in~\Cref{tab:clipping}, the effect of clipping on accuracy depends on the clipping threshold. This is because clipping reduces the $L_2$ sensitivity of PCA from $\nicefrac{12R}{n\delta_{\min}^k}$ to $\nicefrac{12C}{n\delta_{\min}^k}$, allowing for selection of a larger privacy parameter during private learning and better accuracy. However, clipping also introduces inaccuracies in the dataset. 

\begin{table}[h]
\caption{Excess empirical loss of clipped pre-processed DP-pipeline }
\begin{center}
\begin{tabular}{l|c|c|c|c}
\toprule
 & Clipping 0.1 & Clipping 0.7 & Clipping 0.99 & Pre-processed DP pipeline \\
\midrule
1.0 & 0.84 & \textbf{0.89} & 0.88 & \textbf{0.89} \\
2.0 & 0.84 & \textbf{0.89} & \textbf{0.89} & \textbf{0.89} \\
5.0 & 0.85 & \textbf{0.90} & \textbf{0.90} & \textbf{0.90} \\
\bottomrule
\end{tabular}
\end{center}
\label{tab:clipping}
\end{table}

However, we note that $R$ is dataset dependent and might lead to additional privacy leakage.

\end{document}